\documentclass[12pt,psfig]{article}
\usepackage[toc,page]{appendix}
\usepackage{amsmath, amsthm, amssymb,amsfonts,amscd}
\usepackage{graphicx,epsfig,latexsym}
\usepackage{graphicx,epsfig}
\usepackage{amsfonts}
\usepackage{amscd}
\usepackage{mathrsfs}
\usepackage{enumerate}
\usepackage{latexsym}
\def\nalpha{\mathsf{\alpha}}
\def\nbeta{\mathsf{\beta}}
\def\na{a}
\def\nb{b}
\def\nc{c}
\def\nC{\mathsf{C}}
\def\nS{\mathsf{S}}
\def\nN{\mathsf{N}}
\def\nA{\mathsf{A}}
\def\nB{\mathsf{B}}
\def\nfb{\mathfrak{b}}
\def\nfa{\mathfrak{a}}

\def\nX{\mathsf{X}}
\def\nY{\mathsf{Y}}

\newcommand{\Sl}{\mathfrak{sl}_2}

\newcommand{{\N}}{{\rm N}}
\newcommand{{\M}}{{\rm M}}

\newcommand{\ad}{{\rm ad}}

\newcommand{\nt}{\mathsf{t}}
\newcommand{\nv}{\mathsf{v}}
\newcommand{\m}{\mathsf{m}}
\newcommand{\n}{\mathsf{n}}
\newcommand{\h}{\mathsf{h}}
\newcommand{\z}{\mathsf{z}}
\newcommand{\s}{\mathsf{s}}

\newtheorem{thm}{Theorem}[section]

\newtheorem{lem}[thm]{Lemma}

\theoremstyle{definition}
\newtheorem{defn}[thm]{Definition}

\newtheorem{rem}[thm]{Remark}

\numberwithin{equation}{section}
\usepackage{color}
\newcounter{IssueCounter}
\newtheorem{Issue}[IssueCounter]{Issue}

\def\be {\begin{equation}}
\def\ee {\end{equation}}
\def\ba {\begin{eqnarray}}
\def\ea {\end{eqnarray}}
\def\bpr {\begin{proof}}
\def\epr {\end{proof}}
\def\bes {\begin{equation*}}
\def\ees {\end{equation*}}
\def\bas {\begin{eqnarray*}}
\def\eas {\end{eqnarray*}}

\begin{document}
\renewcommand {\thefootnote}{\dag}
\renewcommand {\thefootnote}{\ddag}
\renewcommand {\thefootnote}{ }


\title{Versal Normal Form for Nonsemisimple Singularities}
\author{
Fahimeh Mokhtari\footnote{Corresponding author. Email: fahimeh.mokhtari.fm@gmail.com}
\\Jan A. Sanders \footnote{Email: jan.sanders.a@gmail.com}\\
Department of Mathematics, Faculty of Sciences\\
Vrije Universiteit, De Boelelaan 1081a,\\ 1081 HV Amsterdam, The Netherlands
}

\date{}
\maketitle



\begin{abstract}
The theory of versal normal form has been playing a role in normal form since the introduction of the concept by V.I. Arnol'd in \cite{arnold1971matrices,0036-0279-27-5-A02}.
But there has been no systematic use of it that is in line with the semidirect character of the group of formal transformations on formal vector fields,
that is, the linear part should be done completely first, before one computes the nonlinear terms.
In this paper we address this issue by giving a complete description of a first order calculation in the case of the two- and three-dimensional
irreducible nilpotent cases, which is then followed up by an explicit almost symplectic calculation to find the transformation to versal normal form
in a particular fluid dynamics problem and in the celestial mechanics \(L_4\) problem.
\end{abstract}\vspace{-0.2in}
\vspace{0.10in} \noindent {\it Keywords}: Versal normal form;  \(L_4\) problem;   Nilpotent; \(\Sl\) representation.
\section{Introduction}\label{sec:1}
In normal form theory for general differential equations or symplectic systems around equilibria,
not much attention is usually given to the linear part of the problem.
A typical approach in bifurcation theory is to compute the normal form
of a general system with respect to a given organizing center and add
versal deformation terms (as first considered in \cite{arnold1971matrices,0036-0279-27-5-A02}). One can then analyze all possible bifurcations
in a neighborhood of the organizing center.
While there is nothing wrong with this approach, it does not answer the question where a given system fits in the analysis.
In other words, how does one compute where the given system is in this
neighborhood of the organizing center?

It is this question that we attempt to answer for a number of examples.
Some of these examples will be very concrete, with only one or two parameters to give us a possibility to actually, see the bifurcations,
others are completely general systems where one can use the computation
by just filling in the parameter values of a given system with the same
type of organizing center.

Ideally, before starting the nonlinear computation, the linear system should be brought
in versal normal form in a finite number of steps, as is attempted in \cite{sanders1994versal}.
In practice what one does is to put the linear part in normal form in the same way
as one does the nonlinear part of the equation, but this may involve infinitely many steps.
Since the linear terms influence the computation in every step, this is not
very desirable (contrary to the nonlinear computations, which cannot influence the
linear part unless there is also a constant term to take into account).

In this paper, we address this problem for a very particular system that has been
the subject of several papers already from the versal deformation point of view,
namely the \(L_4\)-problem as described in  \cite{cushman1986versal}.
This paper contains a very clear discussion of the arguments involved in the
versal deformation computation and we will not repeat these here.
The issue we want to address here is to change the infinite series approach
into a finite explicit computation. 
Apart from the \(L_4\)-problem, we have added several examples to illustrate the method
and to show that it is indeed a method, not a computation that happens to work in the one example.
We treat the \(2\)- and \(3\)-dimensional irreducible nilpotent case in section \ref{sec:2Di} and \ref{sec:3D}, respectively.
\markright{{\footnotesize { \hspace{3in}  F. Mokhtari and J.A. Sanders\,\,\,{\it Versal normal form}}}}
\pagestyle{myheadings}
We started this research by computing exponential maps using the generators of the Chevalley normal form of the 
Lie algebra. In the specific \(L_4\)-problem this leads to quartic equations in the flow parameters
and even if one is able to explicitly solve these equations the result is a map full of radical expressions
which will be very hard to use if one applies the result to the full nonlinear problem as is our goal.
We should mention that in the general linear case this does not occur and one can expect that for simply laced simple Lie algebras
this approach will work without problems.

In order to simplify the resulting map that puts the linear system in versal normal form,
we then decided to drop the requirement that the symplectic form be preserved.
As remarked in \cite{koccak1984normal} there is a strong belief that the symplectic form should be preserved,
which is a bit strange if one considers the fact that in order to put the symplectic form in its Darboux normal form,
one has to use (by definition) transformations that are not symplectic.

Dropping this requirement, which has anyway no consequence for the further analysis since we work with the symplectic vector fields,
not with the Hamiltonians, we then proceed as follows.
We first determine a theoretical form of  the versal normal form, depending on a finite number of versal deformation parameters.
Since we want to reach the versal normal form by conjugation, the characteristic polynomial of the original linear vector field and
the versal deformation should be equal. From this equality, we determine the versal deformation parameters (this is in the symplectic case the only nonlinear part of the procedure,
in the general linear case this part is completely straightforward).

Once we have, given a linear vector field \(\nX^\varepsilon_0\),
which consists of an organizing center \(\nX^0_0\) plus terms in a neighborhood of the organizing center,  in order  to compute its versal deformation \(\bar{\nX}^\varepsilon_0\), we need to solve the linear problem \(X_0^\varepsilon T^\varepsilon=T^\varepsilon\bar{X}^\varepsilon_0\) in such a way that \(T^0\) reduces
to the identity and \(\bar{\nX}^\varepsilon_0\) is in versal normal form.
We then can obtain reasonable expressions for the transformation, which can then be put to good use in the nonlinear normal form analysis.
\section{The algorithm}\label{sec:2}
We start with polynomials \(\mathbb{R}[ x_1,\cdots,x_n]\). We then add to these commuting derivations \(\partial_1,\cdots,\partial_n\) 
and consider these as a left \(\mathbb{R}[ x_1,\cdots,x_n]\) module,
such that \([\partial_i ,x_j]=\delta^i_j\). (One could  write \( \partial_i \) as \(\frac{\partial}{\partial_{x_i}}\)).
We write
\(\frac{\partial P}{\partial x_i}\in\mathbb{R}[ x_1,\cdots,x_n]\) for \([\partial_i,P]\). 
We then define a multiplication \(P_i\partial_i\star P_j\partial_j=
P_i \frac{\partial P_j}{\partial x_i} \partial_j\). This defines a non-associative algebra with an associator \(\alpha(x,y,z)=(x\star y)\star z-x\star (y\star z)\) which is symmetric in its first two variables 
(this ensures that the Jacobi identity holds, \cite{MR0161898}) and from it we can define a Lie algebra, the Polynomial Lie algebra by defining the Lie bracket as \([x,y]=x\star y-y\star x\).
Apart possibly from the notation, this is the usual way of defining polynomial vector fields.
We can put a grading on the polynomial vector field by assigning degree \(1\) to the \(x_i\)'s and degree \(-1\) to the \(\partial_j\)'s.
We remark that the \(\star\)-product is a graded product, that is, the degree of \(U \star V\) is the sum of the degree of \(U\) and the degree of \(V,\)
and this makes the Lie algebra into a graded Lie algebra \(\mathfrak{g}=\prod_{k=0}^\infty \mathfrak{g}_k\).
Among the elements in this Lie algebra, as special position is reserved for those of degree zero. They form a Lie subalgebra \(\mathfrak{gl}(n,\mathbb{R})\).

We start with a given linear vector field which we consider as an element in a reductive Lie algebra \(\mathfrak{g}_0\). In our examples, \(\mathfrak{g}_0\) will
be \(\mathfrak{gl}(n,\mathbb{R})\) or \(\mathfrak{sp}(n,\mathbb{R})\).
We chose an organizing center (in all our example this will be characterized by the fact that the real part of all its eigenvalues is zero, since this is where bifurcations happen)
and introduce for organizational reasons a deformation parameter \(\varepsilon\), which at the end of the computation can be set back to \(1\).
 In our first two examples we assume that  the organizing center \(\nX_0^0\) is in real Jordan normal form
as is usually done in normal form theory.
This is not really necessary and might need the knowledge of the spectrum of \(\nX_0^0\), something we try to avoid in this paper, so we stress the fact that the whole construction
works well without this choice. Alternatively one might want to put \(\nX_0^0\) in rational normal form before starting the computation, or not all, as in Section \ref{sec:exampsp4}. All this is a matter of taste and convenience.

We then split \(\nX_0^0\) into a semisimple and nilpotent part, \(\nX_0^0=\s_0+\n_0\), with \(\s_0\) and \(\n_0\) commuting, \(\s_0,\n_0\in\mathfrak{g}_0\).
We remark that this only needs the characteristic polynomial of \(\nX_0^0\) \cite{MR0323842,SVM2007}.
In \(\ker\ad(\s_0)\) (where \(\ad(X)Y=[X,Y]\), as usual) we construct around \(\n_0\) an \(\mathfrak{sl}_2\)-triple \(\langle \n_0 , \h_0 , \m_0\rangle \) as follows, cf. \cite{knapp2013lie}.
Let \(\z_0^\mu\in\mathfrak{g}_0\) be a solution (with free parameter \(\mu\)) of \(\n_0=\ad^2(\n_0) \z_0^\mu\). Put \(\m_0=-2 \z_0^0\) and \(\h_0=[\m_0,\n_0]\). Then
solve \([\h_0,\z_0^\mu]=2 \z_0^\mu\). If \(\z_0^{\mu_0}\) is a solution, put \(\m_0=\z_0^{\mu_0}\) and let \(\h_0=[\m_0,\n_0]\). Then \([\h_0,\n_0]=- 2\n_0\). As in the \(\s_0+\n_0\)-decomposition, this is a completely rational procedure \cite{cushman56nilpotent}. The existence of solutions to the equations
is guaranteed by the Jacobson-Morozow theorem, see \cite{knapp2013lie}.
\begin{rem}
This construction determines
the {\em style} of the normal form, since we will choose \(\ker\ad(\m_0)\) as the complement to \(\mathrm{im}\ \ad(\n_0)\) and {\em costyle} of the normal form transformation, since we will choose
\(\mathrm{im}\ \ad(\m_0)\) as the complement to \(\ker\ad(\n_0)\).
The costyle of  normal form transformation is the way we choose the free parameters in transformations.
 As suggested by the terminology,
other choices of style are also possible and may in specific problems be preferable.
\end{rem}
The versal normal form should be equivalent  to
the {\em rational (or Frobenius) normal form} of the matrix of \(X_0^\varepsilon\), although for that normal form one usually chooses a different style. 
The computation of \(\bar{X}^\varepsilon_0\)
from \(X^\varepsilon_0 \) has already been described.  
The \(T^\varepsilon\) can be computed by linear elimination. 
If for some \(\varepsilon_0\), \(T^\varepsilon\) fails to be invertible, then we should  take \(|Y^\varepsilon_0|<|Y^{\varepsilon_0}_0|\). 
\begin{defn}
	We say that \(\nX_0^\varepsilon=\nX_0^0+\bar\nY_0^\varepsilon\) is in normal form (in \(\Sl\)-style) with respect to \(\nX_0^0\) if \(\bar\nY_0^\varepsilon\in\ker \s_0 \cap \ker \m_0\).
	We say that \(\bar{\nX}_0^\varepsilon=\nX_0^0+\bar\nY_0^\varepsilon\) is a versal normal form with respect to \(\nX_0^0\) if \(\bar\nY^\varepsilon_0\) is in normal form with respect to \(\nX_0^0\) and there exists a 
	\(T^\varepsilon\in \mathrm{GL}(n,\mathbb{R})\) such that \(X^\varepsilon_0 T^\varepsilon=T^\varepsilon\bar{X}^\varepsilon_0\) and \(T_0^0=I\). 
\end{defn}
If the Lie algebra is defined by an invariant bilinear form \(\Omega_0\) (for instance, a symplectic form),
one has to compute the induced form  \(\bar{\Omega}^\varepsilon_0=(T^\varepsilon)^t \Omega_0 T^\varepsilon\). In this case we write
 \(\mathfrak{g}^\Omega\) and \(\mathfrak{g}^{\bar{\Omega}}\). Similar remarks apply to invariant trilinear forms in the less popular (in dynamics) case \(\mathfrak{g}_2\),
the Lie algebra of \(\mathsf{G}_2\), cf. \cite{baez2014𝐺2}, not to be confused with an element of grade two in \(\mathfrak{g}\).
This ensures that the versal deformation vector field 
behaves correctly with respect to \(\bar{\Omega}^\varepsilon_0\), that is, \(\bar{\nX}^\varepsilon_0 \in \mathfrak{g}^{\bar{\Omega}}_0\).
Here we trade symplecticness of the maps involved against computational convenience.
\begin{defn}
	Let \(T^\varepsilon\in{\rm GL(2n,\mathbb{R})}\).
	Then  this induces a  new symplectic form   \(\bar{\Omega}^\varepsilon_0\) and a new vector field \(\bar{X}^\varepsilon_0\) as follows 
	\ba\label{E1}
	\bar{\Omega}^\varepsilon_0&=&(T^\varepsilon)^t \Omega_0 T^\varepsilon,
	\\\label{E2}
	X^\varepsilon_0 T^\varepsilon&=&T^\varepsilon  \bar{X}^\varepsilon_0.
	\ea
\end{defn}
\begin{lem}\label{mainlemma}
	The vector field \(\bar{\nX}^\varepsilon_0\) is \(\bar{\Omega}^\varepsilon_0\)-symplectic iff \(\nX^\varepsilon_0\) is an \(\Omega_0\)-symplectic vector field.
\end{lem}
The claim is that \(\bar{\nX}^\varepsilon_0\) is a \(\bar{\Omega}^\varepsilon_0\)-symplectic vector field, that is, we have  
to prove   that 
\ba
(\bar{X}^\varepsilon_0)^t \bar{\Omega}^\varepsilon_0+\bar{\Omega}^\varepsilon_0\bar{X}^\varepsilon_0=0 .
\ea
\begin{proof}
	Assume \( (X_0^\varepsilon)^t \Omega_0+\Omega_0 X_0^\varepsilon=0\). Then
\bas
&&(\bar{X}^\varepsilon_0)^t \bar{\Omega}^\varepsilon_0+\bar{\Omega}^\varepsilon_0\bar{X}^\varepsilon_0
\\&=&
(\bar{X}^\varepsilon_0)^t (T^\varepsilon)^t \Omega_0 T^\varepsilon+(T^\varepsilon)^t \Omega_0 T^\varepsilon\bar{X}^\varepsilon_0
\\&=&
(T^\varepsilon \bar{X}^\varepsilon_0)^t  \Omega_0 T^\varepsilon+(T^\varepsilon)^t \Omega_0 T^\varepsilon\bar{X}^\varepsilon_0
\\&=&
(X^\varepsilon_0 T^\varepsilon)^t  \Omega_0 T^\varepsilon+(T^\varepsilon)^t \Omega_0 X^\varepsilon_0 T^\varepsilon
\\&=&
(T^\varepsilon)^t \left( (X^\varepsilon_0 )^t \Omega_0 + \Omega_0 X^\varepsilon_0\right) T^\varepsilon
\\&=&0,
\eas
proving the statement of the Lemma.
\end{proof}
The next order step is to compute 
\ba
\exp(\ad({\nt}^\varepsilon_{1}))(\bar{\nX}^\varepsilon_0+\nX^\varepsilon_{1}+\cdots)&=&
\bar{\nX}^\varepsilon_0+\nX^\varepsilon_{1}+[{\nt}^\varepsilon_{1},\bar{\nX}^\varepsilon_0]+\cdots.
\ea
Then we solve
\ba\label{Homological}
\ad(\s_0+\m_0)(\nX^\varepsilon_{1}+[{\nt}^\varepsilon_{1},\bar{\nX}^\varepsilon_0])=0,
\ea
in order to obtain \(\bar{\nX}^\varepsilon_0+\bar{\nX}^\varepsilon_{1}+\cdots\)
  in \(\mathfrak{g}^{\bar{\Omega}}\), or, in the general linear case, in \(\mathfrak{g}\), where \(\nX^\varepsilon_{1}\) is the first order nonlinear term  and, with \({\nt}^\varepsilon_{1}\) a general vector field of order \(1\) and \(\bar{\nX}^\varepsilon_{1}\) is in normal form with respect to \(\nX^0_0\) in the \(\Sl\)-style.

This procedure can then be repeated until the full system is in normal form up to the fixed degree.
The \(\ad(\s_0+\m_0)\) ensures that the normal form will automatically have the \(\Sl\)-style with respect to \(\nX^0_0\).

We should remark here that if we start with a general \({\nt}^\varepsilon_{1}\), there may be free parameters in the normal form corresponding to elements in \(\ker\ad(\s_0)\cap\ker\ad(\n_0)\)  in \({\nt}^\varepsilon_{1}\).
This is analogous to the way unique normal forms are computed \cite{baider1992further,sanders2003normal}. The free parameters may be used to simplify the normal form by removing (typically) higher order \(\varepsilon\)-terms.
There is no style known to us that would be preferable to this simple 
{\em free-costyle}). In  most of our examples the transformation turns out to be in \(\Sl\)-costyle.
\begin{rem}
In some problems, when one wants to do the calculations by hand,
 it pays to view the \(\mathfrak{g}_k\), the
polynomial vector fields as representation spaces of \(\mathfrak{g}_0\),
and more specifically of representation spaces of \(\langle \s_0 , \n_0, \h_0 , \m_0\rangle \). For instance, in \cite{baider1992further} the \(\mathfrak{g}_k\)  is shown to be a direct sum (as vector spaces, not
as Lie algebras) of two irreducible representations of \(\Sl\), \(\nfa_k\) and \(\nfb_k\) and this gives rise to a basis
that is completely natural with respect to the action of the given \(\Sl\) and such that \([\mathfrak{z}_k,\mathfrak{z}_l]\subset \mathfrak{z}_{k+l}\) for \(\mathfrak{z}=\nfa,\nfb\).
\end{rem}
As formulated, the algorithm follows what might be called the {\em rational approach}: no eigenvalues need to be computed, only characteristic polynomials, cf \cite{cushman2017uniform}.
This makes it suitable not only for Computer Algebra Systems, but also for Symbolic Formula Manipulation Systems like FORM \cite{Form00} or FERMAT \cite{lewis2008computer}, which is nice if the problems get big.

An alternative method, which might also work when the vector fields are not finitely generated at any given order and might be called the {\em spectral approach}, is to use the spectrum of \(\s_0\) and \(\h_0\), as is done in the averaging method; we refer for this method to \cite{SVM2007}.
\subsection{Nonlinear nilpotent  versal normal form}
\begin{lem}\label{MAINTHM}
For given  \(\nX_k\in\mathfrak{g}_k,\) \(k>0\), and parametric vector field \(\bar{\nX}_0^\varepsilon=\s_0+\n_0+\bar{\nv}_0^\varepsilon\)
in which \(\bar{
v}_0^\varepsilon\in \ker\,\ad({\m_0})\cap\ker\,\ad({\s_0})\)
there exists a  transformation \(\nt_k^\varepsilon\in \mathfrak{g}_k\) to the following problem 
\bas
\ad(\bar{\nX}_0^\varepsilon)\nt_k^\varepsilon=\nX_k-\bar{\nX}_k^\varepsilon,
\eas 
where \(\bar{\nX}_k^\varepsilon\in \ker\ad({\m_0})\cap \mathfrak{g}_k.\) The transformation \(\nt_k^\varepsilon\) and the normal form \(\bar{\nX}_k^\varepsilon\) can be found explicitly from equations \eqref{transf} and \eqref{nf},
respectively. 
\end{lem}
\bpr
It should be noted that this proof  follows (but with some minor  corrections and clarifications) the proof  given in  \cite[Section 2.3]{sanders1994versal}.

Our problem is that to find the admissible transformation  \(\nt_k^\varepsilon\) and  the obstruction term \(\bar{\nX}_0^\varepsilon\in\ker\ad({\m_0})\cap \mathfrak{g}_k\)  such that the following hold
\bas
\ad(\bar{\nX}_0^\varepsilon)t_k^\varepsilon=\nX_k-\bar{\nX}_k^\varepsilon.
\eas
From \cite[Chapters 11-12]{SVM2007}
  the procedure is given to solve  the following linear problem 
\ba\label{g0}
\ad(\n_0)\nt_k^0=\nX_k-\bar{\nX}_k.
\ea
Denote the transformation \(\nt_k^0\) in equation \eqref{g0} by \(\bar{N}\nX_k.\)
Hence  from the fact that \(V={{\ker}\,\ad(\m_0)}\oplus{\rm{im}\,\ad(\n_0)}\)  one has 
\bas
{\ad (\n_{0})} {\bar N}=\pi_{\rm{im}\,\ad(\n_0)}=1-\pi_{{\ker}\,\ad(\m_0)}.
\eas
Note that  the notation \(\bar{N}\nX_k\) shows that  the operator \(\bar{N}\) acts on  \(\nX_k\).
Let now \(Q=\ad(\s_0+\bar{
v}_0^\varepsilon) {\bar N}\) and \( \hat{Q}={\bar N}\ad(\s_0+\bar{
v}_0^\varepsilon)\).
We will show that \(Q\) and \(\hat{Q}\) are nilpotent operators, so that \((1+Q)^{-1}\) and \((1+\hat{Q})^{-1}\) are both well defined.
Observe that \(\bar{N} Q=\hat{Q}  \bar{N}\).
\begin{lem}
\( \bar{N} (1+Q)^{-1}=(1+\hat{Q})^{-1}\bar{N}\).
\end{lem}
\begin{proof}
We compute
\bas
\bar{N} (1+Q)^{-1}&=&\sum_{i=0}^\infty (-1)^i \bar{N} Q^i = \sum_{i=0}^\infty (-1)^i  \hat{Q}^i \bar{N} = (1+\hat{Q})^{-1}\bar{N},
\eas
and the Lemma is proved.
\end{proof}
We claim that  \(\nt_k^\varepsilon\) is given by
\ba\label{transf}
\nt_k^\varepsilon={\bar N} (1+Q)^{-1} \nX_k=(1+\hat{Q})^{-1}\bar{N}\nX_k=(1+\hat{Q})^{-1}\nt_k^0.
\ea
Therefore we have to first  show that \(Q\) and \(\hat{Q}\) are nilpotent  and \(\nX_k-ad(\bar{\nX}_0^\varepsilon)\nt_k^\varepsilon\in\ker\ad({\m_0})\cap \mathfrak{g}_k.\)
Assume that the  \(\nX_k\) has \(\ad(\h_0)\)-eigenvalue \(\lambda \); then the \(\bar{N}\nX_k\) has  \(\ad(\h_0)\)-eigenvalue \(\lambda+2\) since \(\ad(\h_0)\n_0=-2\n_0\).

By assumption, \(\bar{
v}_0^\varepsilon\in \ker\,\ad({\m_0})\cap\ker\,\ad({\s_0})\); hence \(\ad({\m_0})\bar{
v}_0^\varepsilon=\ad({\s_0})\bar{
v}_0^\varepsilon=0\).
Therefore the \(\ad(\h_0)\)-degree of all terms in \(\bar{
v}_0^\varepsilon\) is \(\ge 0\). Since \(\ad(\h_0)\s_0=\ad(\m_0)\s_0=0\) then its  \(\ad(\h_0)\)-degree is zero. This implies  that the \(\ad(\h_0)\)-degree
of \(Q=\ad(\s_0+\bar{
v}_0^\varepsilon) \bar{N}\ge 2\) hence \(Q\) is nilpotent. The proof for \(\hat{Q}\) is the almost the same.
It follows that  \(1+Q\) and \(1+\hat{Q}\) are invertible. What remains to be done is  to show \(\nX_k-ad(\bar{\nX}_0^\varepsilon)\nt_k^\varepsilon\in\ker\ad({\m_0})\cap \mathfrak{g}_k\):
\bas
\ad(\s_0+\n_0+\bar{
v}_0^\varepsilon)t_k^\varepsilon&=&\ad(\s_0+\n_0+\bar{
v}_0^\varepsilon) {\bar N} (1+Q)^{-1}\nX_k
\\
&=&\left(\ad(\n_0) \bar{N}+Q\right) (1+Q)^{-1}\nX_k
\\
&=&(1+Q-(1-\ad(\n_0){\bar N}))(1+Q)^{-1}\nX_k
\\
&=&
\nX_k-(1-\ad(\n_0) \bar{N})(1+Q)^{-1}\nX_k
\\&=&
\nX_k-\pi_{{\ker}\,\ad(\m_0)}(1+Q)^{-1}\nX_k.
\eas
We rewrite this as
\bas
\nX_k=\pi_{{\ker}\,\ad(\m_0)}(1+Q)^{-1}\nX_k+\ad(\s_0+\n_0+\bar{
v}_0^\varepsilon) {\bar N} (1+Q)^{-1}\nX_k,
\eas
and we define 
\ba\label{nf}
\bar{\nX}^{\varepsilon}_k=\pi_{{\ker}\,\ad(\m_0)} (1+Q)^{-1}\nX_k.
\ea
This concludes the proof of  Lemma \ref{MAINTHM}.
\epr
\subsection{Nonsemisimple versal normal form}\label{sec:Nonsemiversal}
We now extend the versal normal form computation problem from the nilpotent
to the nonsemisimple case.
We follow \cite[Section 2.4]{sanders1994versal}.
We consider the problem
\[
\ad(\nX_0^\varepsilon) \nt_k^\varepsilon =\bar{\nX}_k^\varepsilon-\bar{\bar{\nX}}_k^\varepsilon,\quad  \bar{\nX}_k^\varepsilon\in\ker\ad(\m_0),\quad  \bar{\bar{\nX}}_k^\varepsilon\in\ker\ad(\m_0)\cap
\ker\ad (\s_0).
\]
\def\im{\mathrm{im \ }}
We observe that the right hand side is by definition in \(\ker\ad(\m_0)\cap\im\ad(\s_0)\) and \(\ad(\s_0+\n_0)\) is invertible on this subspace.
We define operators \(K_k:\ker\ad(\m_0)|\mathfrak{g}_k\rightarrow\ker\ad(\m_0)|\mathfrak{g}_k\) such that \(K_k\neq I_{\ker\ad(\m_0)|\mathfrak{g}_k}\)
for \(\varepsilon\neq 0\). Let
\[
K_k=\ad(\bar{\nX}_0^\varepsilon) (1+\hat{Q})^{-1} \pi_{\ker\ad(\n_0)} .
\]
The projection on \(\ker\ad(\n_0) \) is necessary, in order not to interfere with the previous normal form calculation in Section \ref{sec:2}.
We now show that \(K_k:\ker\ad(\m_0)|\mathfrak{g}_k\rightarrow\ker\ad(\m_0)|\mathfrak{g}_k\):
\bas
K_k&=&\ad(\bar{\nX}_0^\varepsilon)(1+\hat{Q})^{-1} \pi_{\ker\ad(\n_0)} 
\\&=&\ad(\bar{\nX}_0^\varepsilon)(1+\hat{Q})^{-1} (1-\bar{N}\ad(\n_0))  
\\&=&\ad(\bar{\nX}_0^\varepsilon)(1+\hat{Q})^{-1} (1+\hat{Q}-\bar{N}\ad(\nX_0^\varepsilon))  
\\&=&\ad(\bar{\nX}_0^\varepsilon) (1-(1+\hat{Q})^{-1}\bar{N}\ad(\nX_0^\varepsilon))  
\\&=&\ad(\bar{\nX}_0^\varepsilon) (1-\bar{N}(1+Q)^{-1}\ad(\nX_0^\varepsilon))  
\\&=& (1-\ad(\bar{\nX}_0^\varepsilon) \bar{N}(1+Q)^{-1})  \ad(\nX_0^\varepsilon)
\\&=& (1-(1-\pi_{\ker\ad(\m_0)}) (1+Q)^{-1}-Q(1+Q)^{-1})  \ad(\nX_0^\varepsilon)
\\&=& (1-(1+Q)^{-1}+\pi_{\ker\ad(\m_0)} (1+Q)^{-1}-Q(1+Q)^{-1})  \ad(\nX_0^\varepsilon)
\\&=& \pi_{\ker\ad(\m_0)} (1+Q)^{-1}  \ad(\nX_0^\varepsilon).
\eas
The map \(\hat{K}_k =K_k \ad^{-1}(\s_0+\n_0)\) is well defined on \(\ker\ad(\m_0)\cap\im\ad(\s_0)|\mathfrak{g}_k\) and reduces to \( 1-\ad(\s_0+\n_0) \bar{N}(1+Q)^{-1}= 1-\ad(\s_0+\n_0)(1+\hat{Q})^{-1}\bar{N}  \) when the perturbation is zero and this reduces to \(1\) on \(\ker\ad(\m_0)\).
This in turn implies that \(\hat{K}_k\) is invertible in a neighborhood of \(\varepsilon=0\),
which means we can find a transformation generator to bring \(\bar{\nX}_k^\varepsilon\) into the normal form \(\bar{\bar{\nX}}_k^\varepsilon\).
The values of \(\varepsilon  \) for which \( \hat{K}_k\) fails to be invertible are called resonances; they play a role in the bifurcation analysis of the \(L_4\)-problem, cf. Section \ref{sec:3body}.

The method we describe here does prove that it is possible to compute the transformation explicitly and if the dimension of \(\mathfrak{g}_k\) is a bit higher,
it may help to reduce the dimension of the linear algebra problem, since one can restrict to \(\ker\ad(\m_0)\).

\section{2D nilpotent -- invariant formulation}\label{sec:2Di}

\subsection{The versal normal form of the linear system}
In this section, we intend to study the versal normal form of two-dimensional nilpotent singularities.
We use this example to illustrate the method in great detail.
This leads at times to statements that sound a bit simplistic;
these are nevertheless stated explicitly so that it is clear what
the flow of the argument is in the later examples, where the complexity of the calculation can obscure what is going on.

Consider  the following two-dimensional perturbed singular system.
\ba\label{2N}
\begin{pmatrix} \dot{x}\\ \noalign{\medskip}\dot{y}\end{pmatrix}
=\begin{pmatrix}\varepsilon \,{\tilde m}_{{1,1}}& \varepsilon {\tilde m}_{1,2}\\ \noalign{\medskip}
	{\tilde m}_{{2,1}}&\varepsilon \,{\tilde m}_{{2,2}}\end{pmatrix} \begin{pmatrix} x\\ \noalign{\medskip}y\end{pmatrix}
=	\tilde{X}^{\varepsilon}_0 \begin{pmatrix} x\\ \noalign{\medskip}y\end{pmatrix},
\ea
where we regard \(\tilde{m}_{{i,j}}\) for all \(i,j=1,2\) as elements of a commutative ring \(R\) of functions of certain parameters taking their
values in \(\mathbb{R}\) (since we want to work with real differential equations) and \({\tilde m}_{2,1}\in {R^*}\) where \(R^*\) denotes  the  invertible elements in the ring \({R}\).
Invertible in this context means that if we use asymptotic estimates, dividing by an {\em invertible} element does not produce
big numbers, which could ruin the asymptotic estimate.
As a consequence one is not allowed to divide by the noninvertible 
elements in the course of the normal form computation.

Since \({\tilde m}_{2,1}\) is invertible,  there exists an invertible  linear transformation 
\bas
T_{(0)}^\varepsilon=\begin{pmatrix}-{\tilde m}^{-1}_{2,1} &0\\ \noalign{\medskip}0&1\end{pmatrix},
\eas 
that takes  \eqref{2N} (with \(	\tilde{X}^{\varepsilon}_0T_{(0)}^\varepsilon=T_{(0)}^\varepsilon {X}^{\varepsilon}_0 \)) to the following 
\ba\label{2Nb}
\begin{pmatrix} \dot{x}\\ \noalign{\medskip}\dot{y}\end{pmatrix}
=\begin{pmatrix} \varepsilon\,{ m}_{{1,1}}& \varepsilon\,{ m}_{{1,2}}\\ \noalign{\medskip}
-1&\varepsilon\,{ m}_{{2,2}}\end{pmatrix} \begin{pmatrix} x\\ \noalign{\medskip}y\end{pmatrix}
= {X}^{\varepsilon}_0 \begin{pmatrix} x\\ \noalign{\medskip}y\end{pmatrix},
\ea
(where \(m_{1,1}=\tilde{m}_{1,1}\), \(m_{2,2}=\tilde{m}_{2,2}\),
and \(m_{1,2}=-\tilde{m}_{1,2}\tilde{m}_{2,1}\)), so that \(- {X}^{0}_0\) is in  Jordan normal form (the minus sign is there to be consistent with the definitions of the \(\nA\) and \(\nB\)-families to follow shortly).

We now rewrite  Equation \eqref{2Nb} to the operator form 
\bas
 \mathsf{\nX}^{\varepsilon}_0=\left(\varepsilon\,{ m}_{{1,1}}x+\varepsilon\,{ m}_{{1,2}}y \right) \frac{\partial}{\partial x}+\left(-x+\varepsilon\,{ m}_{{2,2}} y\right) \frac{\partial}{\partial y},
\eas
and  express \( {\nX}^{\varepsilon}_0\) to the \({\nA}\) and \({\nB}\) families introduced by \cite{baider1992further} (but with \({\nA}\) and \({\nB}\) interchanged) as 
\bas
 {\nX}^{\varepsilon}_0= {\nB}^{1}_0+\frac{\varepsilon}{2}\left({ m}_{{1,1}}+\,{ m}_{{2,2}}\right){\nA}^0_0+\varepsilon\left({\,{ m}_{{1,1}}-\,{ m}_{{2,2}}}\right) {\nB}^0_0+\varepsilon  m_{1,2}{\nB}^{-1}_0.
\eas
We now want (this is the choice of normal form style) \( {\nX}^{\varepsilon}_0-{\nB}^{1}_0\) to  commute with \({\nB}^{-1}_0\);  a general expression of linear vector fields commuting  is \(\varepsilon_{\nB}{\nB}^{-1}_0+ \varepsilon _{\nA} {\nA}^0_0,\) corresponding to the differential equation   
\ba\label{2NS}
\begin{pmatrix} \dot{x}\\ \noalign{\medskip}\dot{y}\end{pmatrix}
=\begin{pmatrix} \varepsilon_{\nA} & \varepsilon_{\nB}\\ \noalign{\medskip}
	-1&\varepsilon_{\nA}\end{pmatrix} \begin{pmatrix} x\\ \noalign{\medskip}y\end{pmatrix}
=\bar{X}^{\varepsilon}_0 \begin{pmatrix} x\\ \noalign{\medskip}y\end{pmatrix},
\ea
(the fact that the \(\varepsilon_{\nA}\) are on the diagonal and will stay there if we go to higher dimensions prompted the interchange of \(\nA\) and \(\nB\) with respect to the definitions
in \cite{baider1992further}) and the differential operator
\ba\label{2NDO}
	\bar{\mathsf{\nX}}^{\varepsilon}_0={\nB}^{1}_0+ \varepsilon_{\nA} {\nA}^0_0+\varepsilon_{\nB}{\nB}^{-1}_0.
\ea
We want to find the transformation that is named \(T_{(1)}^\varepsilon\)  such that  \( {X}^{\varepsilon}_0T_{(1)}^\varepsilon=T_{(1)}^\varepsilon	\bar{X}^{\varepsilon}_0.\) 
The necessary condition  under which such transformation exists is that the characteristic polynomial of \(X^\varepsilon_0\) and \(	\bar{X}^{\varepsilon}_0\)  be the same.
In what follows  using the  characteristic polynomial of \( {X}^{\varepsilon}_0\) and \(	\bar{X}^{\varepsilon}_0\)
we find the  \(\varepsilon_{\nA},\varepsilon_{\nB}.\)   
The characteristic polynomial of \( {X}^{\varepsilon}_0\) and \(	\bar{X}^{\varepsilon}_0\) are given, respectively by
\bas
\chi( {X}^{\varepsilon}_0)&=&{\lambda}^{2}- \varepsilon\left( \,m_{{1,1}}+m_{{2,2}} \right) \lambda+{
	\varepsilon}^{2}m_{{2,2}}m_{{1,1}}+\varepsilon\,m_{{1,2}},
\\
\chi(	\bar{X}^{\varepsilon}_0)&=&{\lambda}^{2}-2\,\varepsilon_{\nA}\lambda+{\varepsilon_{\nA}}^{2}+\varepsilon_{\nB}.
\eas
We define the invariants of \(X^\varepsilon_0\) as
\[
\chi( {X}^{\varepsilon}_0)={\lambda}^{2}-\Delta_1\lambda+\Delta_2.
\]
and we identify \(\Delta_1\) as the trace of \( {X}^{\varepsilon}_0\)
and \(\Delta_2\) as the determinant.
Since the equivalent matrices have the same characteristic polynomial then we find that 
\ba\label{v1n2}
 \varepsilon_{\nA}&=& \frac{1}{2}\Delta_1,
\\\label{v2n2}
 \varepsilon_{\nB}&=&\Delta_2-\frac{1}{4}\Delta^2_1.
\ea
We close this part by  the following theorem (this is not much of a theorem in this particular problem, but we formulate it as such because it is a basic step in this paper).
\begin{thm}
	There exists an invertible transformation \(T_{(1)}^\varepsilon\), defined by
	\ba\label{t2n}
	T_{(1)}^\varepsilon= \begin{pmatrix} 1&0\\ \noalign{\medskip}\frac{\varepsilon}{2}(m_{2,2}-m_{{1,1}})
		&1\end{pmatrix},
	\ea
	which brings  the matrix \eqref{2N} to \eqref{2NS}.
\end{thm}
\bpr
The transformation \eqref{t2n} is obtained using equation \( {X}^{\varepsilon}_0T_{(1)}^\varepsilon=T_{(1)}^\varepsilon	\bar{X}^{\varepsilon}_0\). This is a linear equation in \(T_{(1)}^\varepsilon\) 
and the existence of a solution is shown here explicitly.
\epr
\subsection{Some representation theory }
Following \cite{baider1992further} we  describe vector fields of arbitrary order in a bigraded infinite dimensional Lie algebra \(\nfa\oplus\nfb\),
where \(\nfa\) and \(\nfb\) are bigraded Lie subalgebras and the
\(\oplus\) denotes the direct sum of modules, not of Lie algebras,
as can be seen from the Lie brackets below, and
spanned by elements 
\(
{\nA}_m^n\in\nfa_m , 0\leq n\leq m, 
{\nB}_k^l\in\nfb_k, -1\leq l\leq k+1\)
(i.e. 
\(\dim\nfa_m=m+1\)
and 
\(\dim\nfb_k=k+3\) 
) 
where \({\nA}_{m}^{n}\) and \({\nB}_{k}^{l}\) are defined as 
\ba
{\nA}_{m}^{n}&:=& x^{n}y^{m-n}\left(x\frac{\partial}{\partial x}+y\frac{\partial}{\partial y} \right),\qquad\qquad\qquad\qquad\qquad\qquad\; (0\leq n\leq m),
\\
{\nB}_{k}^{l}& :=&\frac{x^ly^{k-l}}{k+2}\left((k-l+1)x\frac{\partial}{\partial
	x}-(l+1)y\frac{\partial}{\partial y}\right), \qquad\qquad\; (-1\leq l\leq k+1),
\ea
 with brackets 
\ba
{[}{\nA}^l_k,{\nA}^n_m]&=&(m-k)\,{\nA}^{l+n}_{k+m},
\\
{[}{\nB}^l_k,{\nA}^n_m]&=&\frac{m(m+1)}{m+k+2}\left(\frac{n}{m}-\frac{l+1}{k+2}\right)\,{\nA}^{l+n}_{k+m}-k\,{\nB}^{l+n}_{k+m},
\\
{[}{\nB}^l_k,{\nB}^n_m]&=&(k+m+2)\left(\frac{n+1}{m+2}-\frac{l+1}{k+2}\right){\nB}^{l+n}_{k+m}.
\ea
We can now write an arbitrary order \(s\) vector field as
\bas
\nX_s:=
\sum_{l=0}^s {\na}^l_s {\nA}^l_s
+\sum_{l=-1}^{s+1} {\nb}^l_s {\nB}^l_s .
\eas

A general element of order \(s\) in \(\ker\ad({\nB}_0^{-1})\) can be written as
\ba
\bar{\nX}_s=
 \bar{\na}^0_s {\nA}^0_s
+\bar{\nb}^{-1}_s {\nB}^{-1}_s .
\ea

\subsection{Nonlinear normal form reduction}

We now have to solve the equation (in \({t}^\varepsilon_s\))
\ba
\ad({\nB}_0^{-1})(\ad(\bar{\nX}_0^\varepsilon){\nt}^\varepsilon_s-\nX_s)=0,
\ea
where
\ba
{\nt}^\varepsilon_s=
\sum_{l=0}^s \nalpha^l_s {\nA}^l_s
+\sum_{l=-1}^{s+1} \nbeta^l_s {\nB}^l_s .
\ea
Recall that
\bas
\bar{\nX}_s&=&
\bar{\na}^0_s {\nA}^0_s
+\bar{\nb}^{-1}_s {\nB}^{-1}_s,\qquad\qquad
\nX_s
=
\sum_{l=0}^s \na^l_s {\nA}^l_s
+\sum_{l=-1}^{s+1} \nb^l_s {\nB}^l_s .
\eas
We have to solve
\bas
 \bar{\na}^0_s {\nA}^0_s
+\bar{\nb}^{-1}_s {\nB}^{-1}_s &=&\sum_{l=-1}^{s+1} \nb^l_s {\nB}^l_s +\sum_{l=0}^s \na^l_s {\nA}^l_s\\
&&+ \nbeta^{s}_s {\nB}^{s+1}_s
-2 s \varepsilon_{\nA} \nbeta^{s+1}_s {\nB}^{s+1}_s
+ \nalpha^{s-1}_s {\nA}^{s}_s
-2 s \varepsilon_{\nA} \nalpha^s_s {\nA}^s_s
\\&&
-\sum_{k=0}^{s}\left( -(s+2-k) \nbeta^{k-1}_s +2 s \varepsilon_{\nA} \nbeta^k_s -(k+2) \varepsilon_{\nB} \nbeta^{k+1}_s  \right) {\nB}^{k}_s
\\&&
-\sum_{k=1}^{s-1} \left( -(s+1-k) \nalpha^{k-1}_s+2 s \varepsilon_{\nA}  \nalpha^k_s -(k+1) \varepsilon_{\nB} \nalpha^{k+1}_s \right) {\nA}^{k}_s
\\&&
-\left(2 s \varepsilon_{\nA}  \nbeta^{-1}_s -\varepsilon_{\nB}\nbeta^{0}_s \right) {\nB}^{-1}_s
-\left(2 s \varepsilon_{\nA} \nalpha^0_s 
-\varepsilon_{\nB} \nalpha^{1}_s 
\right) {\nA}^{0}_s.
\eas
Thus we find, if we look at the \({\nB}^{s+1}_s\)-term, that
\ba
 \nbeta^{s}_s&=&2 s \varepsilon_{\nA} \nbeta^{s+1}_s-\nb^{s+1}_s , 
 \ea
where \(\nbeta^{s+1}_s\) is a free parameter, to be determined later
at our convenience.

Similarly, looking at the \({\nA}^s_s\) terms we find
\ba
\nalpha^{s-1}_s &=&
2 s \varepsilon_{\nA} \nalpha^s_s -
\na^{s}_s  ,
\ea
where \(\nalpha^{s}_s\) is the free parameter.
For \(0\leq k \leq s\) we find, looking at the \({\nB}^k_s\),
\ba
(s+2-k) \nbeta^{k-1}_s &=&2 s \varepsilon_{\nA} \nbeta^k_s -(k+2) \varepsilon_{\nB} \nbeta^{k+1}_s
-  \nb^{k}_s.
\ea
For \(1\leq k \leq s-1\) we find, looking at the \({\nA}^k_s\),
\ba
 (s+1-k) \nalpha^{k-1}_s=2 s \varepsilon_{\nA}  \nalpha^k_s -(k+1)\varepsilon_{\nB}  \nalpha^{k+1}_s
- \na^{k}_s.
\ea
Then
\ba
\bar{\nX}_s&=&
\left( \nb^{-1}_s-2 s \varepsilon_{\nA}  \nbeta^{-1}_s +\varepsilon_{\nB}\nbeta^{0}_s \right) {\nB}^{-1}_s
+\left(\na^{0}_s  -2 s \varepsilon_{\nA} \nalpha^0_s 
+\varepsilon_{\nB} \nalpha^{1}_s 
 \right) {\nA}^{0}_s.
\ea
Let us now specialize to \(s=1\).
We find
\bas
\nalpha^{0}_1 &=&
2 \varepsilon_{\nA} \nalpha^1_1-
\na^{1}_1,
\\
 \nbeta^{1}_1&=&2 \varepsilon_{\nA}\nbeta^{2}_1-\nb^{2}_1,\\ 
\nbeta^{0}_1 
&=&\frac{1}{2}\left( 4\varepsilon^2_{\nA}  -3 \varepsilon_{\nB}\right) \nbeta^{2}_1
-\frac{1}{2}\nb^{1}_1-\varepsilon_{\nA}\nb^{2}_1,
\\
 \nbeta^{-1}_1 &=& \frac{1}{3} \varepsilon_{\nA}\left( 4\varepsilon^2_{\nA} -7\varepsilon_{\nB}\right) \nbeta^{2}_1
+\frac{2}{3} (\varepsilon_{\nB}-\varepsilon^2_{\nA}) \nb^{2}_1
- \frac{1}{3}\left(\varepsilon_{\nA}\nb^{1}_1+  \nb^{0}_1\right),
\eas
and
\bas
\bar{\nX}_1&=&
\left( 
\nb^{-1}_1
-2 \varepsilon_{\nA}  \nbeta^{-1}_1
 +\varepsilon_{\nB}\nbeta^{0}_1 
 \right) {\nB}^{-1}_1
+\left(\na^{0}_1 -2 \varepsilon_{\nA}\nalpha^0_s 
+\varepsilon_{\nB} \nalpha^{1}_1
\right) {\nA}^{0}_1
\\&=&\left( 
\nb^{-1}_1+\frac{2}{3}\varepsilon_{\nA}\nb^{0}_1
- (-\frac{4}{3}\varepsilon^2_{\nA}+\varepsilon_{\nB})\frac{1}{2}\nb^{1}_1
+\frac{7}{3} ( \varepsilon_{\nB}
-\frac{4}{7}\varepsilon^2_{\nA}) \varepsilon_{\nA}\nb^{2}_1 
\right.\left.
+(\frac{20}{3}   \varepsilon^2_{\nA} \varepsilon_{\nB}
-\frac{8}{3} \varepsilon^4_{\nA}  
 -\frac{3}{2}\varepsilon^2_{\nB}) \nbeta^{2}_1
\right) {\nB}^{-1}_1
\\&&
+\left(\na^{0}_1 +2 \varepsilon_{\nA}\na^{1}_1 
+(-4\varepsilon^2_{\nA} +\varepsilon_{\nB}) \nalpha^{1}_1
\right) {\nA}^{0}_1.
\eas 
Choosing \(\nbeta^2_1=0\) and \(\nalpha^1_1=0\) (in accordance with the \(\Sl\)-costyle) we find
\def\Tr{\mathrm{Tr\ }}
\def\Det{\mathrm{Det\ }}	\ba
\bar{\nX}_1&=&\left( 
\nb^{-1}_1+\frac{2}{3}\varepsilon_{\nA}\nb^{0}_1
-\frac{1}{2} (-\frac{4}{3}\varepsilon^2_{\nA}+\varepsilon_{\nB})\nb^{1}_1
+\frac{7}{3} ( \varepsilon_{\nB}
-\frac{4}{7}\varepsilon^2_{\nA}) \varepsilon_{\nA}\nb^{2}_1 
\right) {\nB}^{-1}_1\nonumber
+\left(\na^{0}_1 +2 \varepsilon_{\nA}\na^{1}_1
\right) {\nA}^{0}_1\\
&=&\left( 
\nb^{-1}_1+\frac{1}{3}(\Tr X_0^\varepsilon) \nb^{0}_1
+ \frac{1}{2}(\frac{1}{12}\Tr^2 X_0^\varepsilon+\Det X_0^\varepsilon )\nb^{1}_1
+\frac{7}{6} (\frac{3}{28} \Tr^2 X_0^\varepsilon -\Det X_0^\varepsilon) (\Tr X_0^\varepsilon)  \nb^{2}_1 
\right) {\nB}^{-1}_1\nonumber
\\&&
+\left(\na^{0}_1 +(\Tr X_0^\varepsilon) \na^{1}_1
\right) {\nA}^{0}_1.\label{form:universal}
\ea

It follows from the definitions in \cite{baider1992further} that 
\ba
{\nA}^0_1&=&\iota_2(\na_0^1)=\iota_2(x)=x\begin{pmatrix} x\\y \end{pmatrix},
\\
{\nB}^{-1}_1&=&\frac{1}{3}\iota_1(\nb^{-1}_1)=\frac{1}{3}\iota_1(x^3)=
x^2\begin{pmatrix}0\\1
\end{pmatrix}.
\ea
\section{3D irreducible nilpotent}\label{sec:3D}
  \subsection{The versal normal form of the linear system}
  In this section, we discuss versal deformation of  three-dimensional nilpotent singularities.
  Consider the deformed nilpotent system 
  \ba\label{3N}
  \begin{pmatrix} \dot{x}\\ \noalign{\medskip}\dot{y}\\ \noalign{\medskip}\dot{z} \end{pmatrix}
  =\begin{pmatrix}\varepsilon\,{\tilde m}_{{1,1}}& \varepsilon{\tilde m}_{1,2}&\varepsilon\,{\tilde m}_{{1,3}}
  	\\ \noalign{\medskip}{\tilde m}_{{2,1}}&\varepsilon\,{\tilde m}_{{2,2}}&\varepsilon{\tilde m}_{2,3}
  	\\ \noalign{\medskip}\varepsilon\,{\tilde m}_{{3,1}}&{\tilde m}_{{3,2}}&\varepsilon
  	\,{\tilde m}_{{3,3}}\end{pmatrix} \begin{pmatrix} x\\ \noalign{\medskip}y\\\noalign{\medskip}z\end{pmatrix}=  \tilde{X}^\varepsilon_0\begin{pmatrix} x\\ \noalign{\medskip}y\\\noalign{\medskip}z\end{pmatrix},
  \ea
   where the elements \({\tilde m}_{2,1}, {\tilde m}_{3,2}\in {R^*}.\) 
      By applying the following invertible transformation   
  \ba\label{T3}
    T_{(0)}^\varepsilon&=&\begin{pmatrix} \tilde{m}_{2,1}&0&0\\ \noalign{\medskip}0&-\tilde{m}_{2,1}^2&0
  \\ \noalign{\medskip}0&-\varepsilon\,\tilde{m}_{{3,1}}{\tilde{m}_{21}}&\alpha
\end{pmatrix}, 
\ea
where
\[\alpha=-\frac{1}{2}{\varepsilon}^{2}\tilde{m}_{{3,1}} \left( \varepsilon\,\tilde{m}_{{2,3}}\tilde{m}_{{3,1}}+{\tilde{m}_{2,1}}
\,\tilde{m}_{{2,2}} \right) +\frac{1}{2}{\tilde{m}_{2,1}} \left( {\varepsilon}^{2}\tilde{m}_{{3,1}}\tilde{m}_{{3,3
}}+{\tilde{m}_{2,1}}\,{\tilde{m}_{3,2}} \right),\]
   the system \eqref{3N} transforms to the following system
  \ba\label{3NN}
  \begin{pmatrix} \dot{x}\\ \noalign{\medskip}\dot{y}\\ \noalign{\medskip}\dot{z} \end{pmatrix}
  = \begin{pmatrix} \varepsilon\,m_{{1,1}}&\varepsilon\,m_{{1,2}}&\varepsilon\,m_{{1,3}}
  	\\ \noalign{\medskip}-1&\varepsilon\,m_{{2,2}}&\varepsilon\,m_{{2,3}}
  	\\ \noalign{\medskip}0&-2&\varepsilon
  	\,m_{{3,3}}\end{pmatrix} \begin{pmatrix} x\\ \noalign{\medskip}y\\\noalign{\medskip}z\end{pmatrix}=   X^\varepsilon_0\begin{pmatrix} x\\ \noalign{\medskip}y\\\noalign{\medskip}z\end{pmatrix},
  \ea
  in which
  \bas
  m_{1,1}&=&  {\tilde m}_{1,1},
  		\\
  		m_{1,2}&=&
-{\varepsilon}^{2}\tilde{m}_{{1,3}}\tilde{m}_{{3,1}}-{\tilde{m}_{2,1}}\,\varepsilon\,\tilde{m}_{{1,2}},
\\
m_{1,3}& =&
\frac {\tilde{m}_{{1,3}}\alpha}{\tilde{m}_{2,1}},
		\\
		m_{2,2}&=&
	{\frac {\varepsilon\, \left( \varepsilon\,\tilde{m}_{{2,3}}\tilde{m}_{{3,1}}+{\tilde{m}_{2,1}}\,m_{{
					2,2}} \right) }{{\tilde{m}_{2,1}}}},
					\\
					m_{2,3}&=&
				-\frac {\tilde{m}_{{2,3}}\alpha}{\tilde{m}_{2,1}^2},
						\\
						m_{3,3}&=&{\frac {\varepsilon\, \left( -\varepsilon\,\tilde{m}_{{2,3}}\tilde{m}_{{3,1}}+{\tilde{m}_{2,1}}\,m_{
									{3,3}} \right) }{{\tilde{m}_{2,1}}}}
						.
							  \eas
							  \begin{rem}
							  	Note that   due to  the assumption 
							  	\({\tilde m}_{2,1}, {\tilde m}_{3,2}\in {R^*}\) 
							  	the transformation   given by \eqref{T3}
							  	when  \(\varepsilon=0\)  is invertible.   
							  \end{rem}
   Now,   we  writing down  \eqref{3NN} in terms of vector fields from \(\mathscr{A},\mathscr{B},\mathscr{C}\) given in  \cite{gazor2017vector} to find 
 the  following 
   \bas
   {\nX}^{\varepsilon}_0&=& {\nB}^{1}_{0,0}+\varepsilon\,m_{{1,3}}{\nC}^{-2}_{0,0}+\frac{1}{2}\varepsilon\, \left( m_{{1,2}}-2\,m_{{2,3}} \right){\nC}^{-1}_{0,0}-\frac{1}{2}\varepsilon\,m_{{2,2}}{\nC}^{0}_{0,0}+\frac{1}{4}\varepsilon\, \left( m_{{1,2}}+2\,m_{{2,3}} \right){\nB}^{-1}_{0,0}
   \\
  &+& \frac{1}{2}\varepsilon\, \left( 2\,m_{{1,1}}+m_{{2,2}} \right) {\nB}^{0}_{0,0}+\varepsilon\, \left( m_{{1,1}}+m_{{2,2}}+m_{{3,3}} \right) \mathsf{A}^0_{0,0}.
   \eas
   Due to \(\Sl\)-style normal form, in order to find the versal normal form of \({\nX}^{\varepsilon}_0\) we seek  the  vector fields which belong to  \(\ker \ad(\nB^{-1}_{0,0}).\)
      Hence the following special structure constants   associated to the \(\nB^{-1}_{0,0}\) are given
  \bas
{[{\nB}^{-1}_{0,0},{\rm  B}_{i,k}^l]}&=&(l+1){\nB}^{l-1}_{i,k},
\\
{[{\nB}^{-1}_{0,0},\mathsf{A}_{i,k}^l]}&=&l\mathsf{A}^{l-1}_{i,k},
\\
{[{\nB}^{-1}_{0,0},{\nC}_{i,k}^l]}&=&\frac{(l+2)(2i+3-l)}{(2i-l+1)}{\nC}^{l-1}_{i,k},\qquad  \hbox{for}\,\,\,  l<2i+1,
\\
{[{\nB}^{-1}_{0,0},{\nC}_{i,k}^l]}&=&0,\qquad\qquad \qquad\qquad \hbox{for}\,\,\,  l=2i+1,
\\
{[{\nB}^{-1}_{0,0},{\nC}_{i,k}^l]}&=&(2i+4) {\nC}^{2i+1}_{i,k},\qquad  \hbox{for}\,\,\,  l=2i+2.
   \eas
  Therefore we  obtain that 
   \ba\label{2NDObar}
   \bar{\mathsf{\nX}}^{\varepsilon}_0={\nB}^{1}_{0,0}+\varepsilon_{\nA} \mathsf{A}^0_{0,0}+\varepsilon_{\nB}{\nB}^{-1}_{0,0}+ \varepsilon_{\nC} {\nC}^{-2}_{0,0},
   \ea
 and    the correspondence differential equation of \(\bar{\mathsf{\nX}}^{\varepsilon}_0\) is
  \ba\label{3NNM}
  \begin{pmatrix} \dot{x}\\ \noalign{\medskip}\dot{y}\\ \noalign{\medskip}\dot{z} \end{pmatrix}
  = \begin{pmatrix} \varepsilon_{\nA}&2\varepsilon_{\nB}&\varepsilon_{\nC}
  	\\ \noalign{\medskip}-1&\varepsilon_{\nA}&\varepsilon_{\nB}
  	\\ \noalign{\medskip}0&-2&\varepsilon_{\nA}
  	\end{pmatrix} \begin{pmatrix} x\\ \noalign{\medskip}y\\\noalign{\medskip}z\end{pmatrix}=  {\bar X}^{\varepsilon}_0\begin{pmatrix} x\\ \noalign{\medskip}y\\\noalign{\medskip}z\end{pmatrix}.
  \ea
  Now we are ready to find  the versal parameters \(\varepsilon_{\nA},\varepsilon_{\nB}\) and \(\varepsilon_{\nC}.\)
  As before by computing  the 
characteristic polynomial of \(\bar{X}^\varepsilon_0\) and   \(\bar{X}^\varepsilon_0\)  we get
  \bas
 \chi(\bar{X}^\varepsilon_0)&=&{\lambda}^{3}-3\,\varepsilon_{{\nA}}{\lambda}^{2}+ \left( 3\,\varepsilon_{{\nA}}^{2}+4\,
 \varepsilon_{{\nB}} \right) \lambda-\varepsilon_{{\nA}}^{3}-4\,\varepsilon_{{\nA}}
 \varepsilon_{{\nB}}-2\,\varepsilon_{{\nC}},
   \\
   \chi({X}^\varepsilon_0 )&=&{\lambda}^{3}-\varepsilon\, \left( m_{{1,1}}+m_{{2,2}}+m_{{3,3}} \right) {\lambda}^{2
   }+\varepsilon\, \big( \varepsilon\,m_{{1,1}}m_{{2,2}}
+\varepsilon\,m_{{1,1}}m_
   {{3,3}}+\varepsilon\,_{{2,2}}m_{{3,3}}
   \\&&
   +m_{{1,2}}+2\,m_{{2,3}} \big) \lambda
      -
   \varepsilon\, \left( {\varepsilon}^{2}m_{{1,1}}m_{{2,2}}m_{{3,3}}+2\,
   \varepsilon\,m_{{1,1}}m_{{2,3}}+\varepsilon\,m_{{1,2}}m_{{3,3}}+2\,m_{{1,3}}
   \right). 
     \eas
  These define the invariants \(\Delta_i, i=1,2,3\) by
  \[
   \chi( X^\varepsilon_0)=\lambda^3-\Delta_1 \lambda^2 +\Delta_2 \lambda -\Delta_3.
  \]
  Hence we find 
  \bas
  \varepsilon_{\nA}&=& \frac{1}{3}\Delta_1,
 \\
 \varepsilon_{\nB}&=&\frac{1}{4}\Delta_2-\frac{1}{12}\Delta^2_1,
 \\
  \varepsilon_{\nC}&=&\frac{1}{2}\Delta_3-\frac{1}{6}\Delta_1\Delta_2+\frac{1}{27}\Delta^3_1.
  \eas
   \begin{thm}\label{Transformation3d}
  	There exists invertible transformation \(	T_{(1)}^\varepsilon\) as 
  	\ba\label{t3n}
 	T_{(1)}^\varepsilon=   \begin{pmatrix} 1&\varepsilon t_1&\varepsilon t_2\\ \noalign{\medskip}0&1&\varepsilon t_3\\ \noalign{\medskip}0&  0&1\end {pmatrix},
  	\ea 
    	in which
  	\bas
  t_1&=& \frac{1}{3}\left(m_{{3,3}}+m_{{2,2}}-2\,m_{{1,1}}\right),
  	\\
   	t_2&=&\frac{5}{36}{{\varepsilon\,m_{{1,1}} \left( m_{{1,1}}-m_{{2,2}}-m_{{3,3}}
   			\right) }}-\frac{1}{36}\,\varepsilon\,m_{{2,2}} \left(m_{{2,2	}}-7\,m_{{3,3}} \right) -\frac{1}{36}\,\varepsilon\,{m_{{3,3}}}^{2}
   \\&&
   +\frac{1}{2}\,m_{{2,3}}-\frac{1}{4}\,m_{{1,
   			2}},
   	  \\
  	 t_3&=&\frac{1}{3}\left(m_{{3,3}}-\frac{1}{2}\,m_{{2,2}}-\frac{1}{2}\,m_{{1,1}}\right),
  	\eas     
  	which brings  the matrix \eqref{3NN} to \eqref{3NNM}.
  \end{thm}
  \bpr
  The transformation \(T_{(1)}^\varepsilon\) is obtained using equation \(X^\varepsilon_0 T_{(1)}^\varepsilon=T_{(1)}^\varepsilon{\bar X}^\varepsilon_0.\) 
  \epr
\subsection{Quadratic nonlinear versal normal form of triple-zero}
In this part we shall  compute the nonlinear normal form of the  following parametric  vector fields  with triple zero bifurcation point 
\ba\label{triple}
\bar{\mathsf{\nX}}^{\varepsilon}_0+\nX_1,
\ea
where 
\ba\label{nnl}
\nX_1&=&\sum_{j=0}^{2} a_j \mathsf{A}^j_{1,0}+\sum_{j=-1}^{3} b_j{\nB }^{j}_{1, 0}+\sum^{0}_{j=-2} c^j_{ 1} {\nC}^{j}_{-1, 1}+\sum^{4}_{j=-2} c^j_{ 0} {\nC}^{j}_{1, 0},
\ea
or equivalently 

\ba\label{nnl1}
\nX_1&=&\sum_{i+j+k=2}  a^{(1)}_{{i,j,k}} x_1^i x_2^jx_3^k \frac{\partial}{\partial x_1}+\sum_{i+j+k=2}  a^{(2)}_{{i,j,k}} x_1^i x_2^jx_3^k \frac{\partial}{\partial x_1}+\sum_{i+j+k=2}  a^{(3)}_{{i,j,k}} x_1^i x_2^jx_3^k \frac{\partial}{\partial x_3}.
\ea
The coefficients of \eqref{nnl} and \eqref{nnl1}  are related by  these  relations:
\bas
a^{(1)}_{{0,0,2}}&=&c^{-2}_1,\quad\quad\quad\quad
a^{(1)}_{{2,0,0}}=\frac{1}{2}\,b_2+c^{2}_0+a_2,\quad\quad
a^{(1)}_{{0,2,0}}=b_0-c^{-2}_{1}+\frac{2}{3}\,c^{0}_{0},\quad
\\
a^{(2)}_{{0,2,0}}&=&c^{-1}_1-c^{1}_0+a_1,\quad\quad\quad
a^{(2)}_{{0,0,2}}=b_{-1}-\frac{1}{4}\,c^{-1}_{0},\quad\quad\quad
a^{(3)}_{{2,0,0}}=30\,c^{4}_{0},\quad
\\
a^{(3)}_{{0,0,2}}&=&-\frac{1}{2}\,b_0+\frac{1}{6}\,c^{0}_{0}+a_0,\quad
a^{(1)}_{{1,1,0}}=b_1+c^{1}_0+a_1,\quad\quad
a^{(1)}_{{1,0,1}}=\frac{1}{3}\,c^{0}_{0}+\frac{1}{2}\,b_0+a_0+
a^{-2}_{1},\quad
\\
a^{(2)}_{{1,1,0}}&=&-\frac{1}{2}\,b_2-4\,c^{2}_0+a_2,\quad\quad
a^{(2)}_{{1,0,1}}=-\frac{1}{2}\,c^{1}_0-c^{-1}_1,\quad\quad
a^{(2)}_{{0,1,1}}=-\frac{2}{3}\,c^{0}_{0}+\frac{1}{2}\,b_0+a_0,\quad
\\
a^{(3)}_{{1,1,0}}&=&-2\,b_3+20\,c^3_0,\quad\quad
a^{(3)}_{{1,0,1}}=-\frac{1}{2}\,b_2+2\,c^{2}_0+a_2+2
\,a^{0}_{1},\quad\quad c^{(1)}_{{0,1,1}}=2\,b_{-1}+c^{-1}_{0},
\\
a^{(3)}_{{0,1,1}}&=&-b_1+a_1+c^{1}_0,\quad\quad
a^{(3)}_{{0,2,0}}=-b_2-2\,c^{0}_{1}+4\,c^{2}_0,\quad\quad c^{(2)}_{{2,0,0}}=-b_3-5\,c^3_0.\quad
\eas
\begin{thm}
The normal form of  \eqref{triple}  
is given by 
\ba\label{2NDO2}
\bar{\mathsf{\nX}}^{\varepsilon}_0={\nB}^{1}_{0,0}+\varepsilon_{\nA} \mathsf{A}^0_{0,0}+\varepsilon_{\nB}{\nB}^{-1}_{0,0}+ \varepsilon_{\nC} {\nC}^{-2}_{0,0}+\bar{\nc}^0_1 {\nC}^{-2}_{1,0}+\bar{\nc}^0_{-1} {\nC}^{-2}_{-1,1}
+\bar{\nb}^{1}_1 {\nB}^{-1}_{1,0}+\bar{a}^{1}_0\mathsf{A}^0_{1,0},
\ea
or equivalently 
\bas
\bar{\nX}_1&=&\left(x_{{3}}\varepsilon_{{\nC}}+2\,x_{{2}}\varepsilon_{{\nB}}+x_{{1}}\varepsilon_{{\nA}}+\bar{\nc}^0_1{x_{{3}}}^{2}+\bar{\nc}^0_{-1} \left( x_{{1}}x_{{3}}-{x_{{2}}}^{2}
\right)+2\bar{\nb}^{1}_1x_{{2}}x_{{3}}+\bar{a}^{1}_0x_{{3}}x_{{1}}
\right) \frac{\partial}{\partial x_1}
\\&&
+
\left(-x_1+\varepsilon_{{\nA}}x_{{2}}+\varepsilon_{{\nB}}x_{{3}}+\bar{\nb}^{1}_1 {x_{{3}}}^{2} +\bar{a}^{1}_0x_{{3}}x_2\right) \frac{\partial}{\partial x_2}
+\left(-2\,x_{{2}}+\varepsilon_{{\nA}}x_{{3}}+\bar{a}^{1}_0x_{{3}}^2\right)\frac{\partial}{\partial x_3},
\eas
where 
\bas
5\bar{\nc}^0_{-1}&:=&{3}\,a^{(1)}_{{1,0,1}}-{2}\,a^{(3)}_{{0,0,2}}
-a^{(2)}_{{0,1,1}}
-{2}\,a^{(1)}_{{0,2,0}}
-\left(
\frac{5}{2}\,a^{(3)}_{{1,1,0}}
+{6}\,a^{(2)}_{{2,0,0}}
\right)\varepsilon_{\nC}
-\Big(
{2}\,a^{(1)}_{{2,0,0}}
-{3}\,a^{(3)}_{{1,0,1}}
+{2}\,a^{(3)}_{{0,2,0}}
\\&&
+a^{(2)}_{{1,1,0}}
\Big)\varepsilon_{\nB}
-\left(
\frac{1}{2}\,a^{(1)}_{{1,1,0}}
-a^{(2)}_{{0,2,0}}
+{4}\,a^{(2)}_{{1,0,1}}
+a^{(3)}_{{0,1,1}}
\right)\varepsilon_{\nA}
+\left(
{4}a^{(2)}_{{2,0,0}}
+{3}a^{(3)}_{{1,1,0}}
\right)\varepsilon_{\nB}\varepsilon_{{\nA}}
\\&&
-\frac{5}{2}\,a^{(3)}_{{2,0,0}} \varepsilon_{\nC} \varepsilon_{\nA}
-
\left(
\frac{1}{2}\,a^{(2)}_{{1,1,0}}
-\frac{3}{2}\,a^{(3)}_{{1,0,1}}
+\,a^{(3)}_{{0,2,0}}
+a^{(1)}_{{2,0,0}}
\right) \varepsilon_{\nA}^2
+ 
\left({\frac {3}{4}\,a^{(3)}_{{1,1,0}}}
+a^{(2)}_{{2,0,0}}\right)\varepsilon_{\nA}^3
\\&&
+{3}\,a^{(3)}_{{2,0,0}}\varepsilon_{\nB}  \varepsilon_{\nA}^2
+{\frac {3}{4}\,a^{(3)}_{{2,0,0}}} \varepsilon_{\nA}^4,
\eas
\bas
10\bar{a}^{1}_0&:=&
\left(
\,a^{(1)}_{{0,2,0}}
+\,a^{(1)}_{{1,0,1}}
+{3}\,a^{(2)}_{{0,1,1}}
+6\,a^{(3)}_{{0,0,2}}
\right)
+\left(-2\,a^{(2)}_{{2,0,0}}
+\frac{5}{2}\,a^{(3)}_{{1,1,0}}\right) \varepsilon_{\nC}
\\&&
+\left(6a^{(1)}_{{2,0,0}}
+{3}a^{(2)}_{{1,1,0}}
+a^{(3)}_{{0,2,0}}
+a^{(3)}_{{1,0,1}}\right)\varepsilon_{\nB}
+
2\left({\frac {3\,}{4}}a^{(1)}_{{1,1,0}}
+\,a^{(2)}_{{0,2,0}}
+\,a^{(2)}_{{1,0,1}}
+{\frac 
{3\,}{4}}a^{(3)}_{{0,1,1}}
\right)\varepsilon_{\nA}
\\&&
+\frac{5}{2}\,a^{(3)}_{{2,0,0}} \varepsilon_{\nC}\varepsilon_{\nA}
+\left(
-12\,a^{(2)}_{{2,0,0}}
+\,a^{(3)}_{{1,1,0}}
\right)\varepsilon_{\nA}\varepsilon_{\nB}
+\left(
-{3}a^{(2)}_{{2,0,0}}
+\frac{1}{4}\,a^{(3)}_{{1,1,0}}
\right)\varepsilon_{\nA}^3+\,a^{(3)}_{{2,0,0}}\varepsilon_{\nA}^2 \varepsilon_{\nB}
\\&&
+\frac{1}{4}\,a^{(3)}_{{2,0,0}}\varepsilon_{\nA}^4,
\\
6\bar{\nb}^{1}_1 &:=&
\,a^{(1)}_{{0,1,1}}
+4\,a^{(2)}_{{0,0,2}}
+\frac{1}{18} \left( 
4\,a^{(1)}_{{2,0,0}}
+\,a^{(2)}_{{1,1,0}}
+\,a^{(3)}_{{0,2,0}}
-{5 \,a^{(3)}_{{1,0,1}}} 
\right) \varepsilon_{\nC}
+ \left( 
\,a^{(1)}_{{1,1,0}}
-\,a^{(3)}_{{0,1,1}} 
\right) \varepsilon_{\nB}
\\&&
+\frac{1}{2} \left( 
\,a^{(1)}_{{0,2,0}}
+\, a^{(1)}_{{1,0,1}}
+\,a^{(2)}_{{0,1,1}}
-2\,a^{(3)}_{{0,0,2}} 
\right) \varepsilon_{\nA}
-\frac{1}{12} \left( 
{ {28}\,a^{(2)}_{{2,0,0}}}
+{ {19}\,a^{(3)}_{{1,1,0}}} 
\right) \varepsilon_{\nC}\varepsilon_{\nA}-\frac {10}{3}\,a^{(3)}_{{2,0,0}}\varepsilon_{\nB}\varepsilon_{\nC} 
\\&&
- \left(
4\,a^{(2)}_{{2,0,0}}
+\,a^{(3)}_{{1,1,0}} 
\right) \varepsilon_{\nB}^{2}
+\frac {5}{6} \left(
2\,a^{(1)}_{{2,0,0}}
-\,a^{(2)}_{{1,1,0}}
-\,a^{(3)}_{{0,2,0}}
-\,a^{(3)}_{{1,0,1}} 
\right)\varepsilon_{\nA} \varepsilon_{\nB}
+ \frac{1}{4}\left(
\,a^{(1)}_{{1,1,0}}
-\,a^{(3)}_{{0,1,1 }} 
\right) \varepsilon_{\nA}^{2}
\\&&
-\frac{2}{3} \left( 
4\,a^{(2)}_{{2,0,0}}
+\,a^{(3)}_{{1,1,0}} 
\right) \varepsilon_{\nA}^{2}\varepsilon_{\nB}-\frac{4}{9}\,a^{(3)}_{{2,0,0}}\varepsilon_{\nA}{\varepsilon_{{\nB}}}^{2}
+\frac{1}{12}\left( 
2\,a^{(1)}_{{2,0,0}}
- {b_{1,1,0}}
- {a^{(3)}_{{0,2,0}}}
- {a^{(3)}_{1,0,1}}
\right) \varepsilon_{\nA}^{3}
\\&&
-\frac {25}{12}\,a^{(3)}_{{2,0,0}} \varepsilon_{\nA}^{2}\varepsilon_{\nC}
-\frac {5}{6}\,a^{(3)}_{{2,0,0}}\varepsilon_{\nA}^{3}\varepsilon_{\nB}
-{\frac{1}{24}}\left(
 4\,a^{(2)}_{{2,0,0}}
+ {a^{(3)}_{{1,1,0}}} 
\right) \varepsilon_{\nA}^{4}
-\frac{1}{24} {a^{(3)}_{{2,0,0}}\varepsilon_{\nA}^{5}},
\\\nonumber
\bar{\nc}^0_1&:=&
a^{(1)}_{{0,0,2}}
+ \frac{1}{12}\left( 
\,a^{(1)}_{{1,1,0}}
+2\,a^{(2)}_{{0,2,0}}
-10\,a^{(2)}_{{1,0,1}}
-3\,a^{(3)}_{{0,1,1}}
\right) \varepsilon_{\nC}+ \frac{1}{3}\left( \frac{1}{2}\,a^{(1)}_{
	{0,1,1}}-\,a^{(2)}_{{0,0,2}} \right) \varepsilon_{\nA}-\frac{1}{3}\,a^{(3)}_{{2,0,0}}{ \varepsilon_{\nC}}^{2}
\\&&
+ \frac{1}{5}\left( 
\,a^{(1)}_{{0,2,0}}
+\,a^{(1)}_{{1,0,1}}
-2 \,a^{(2)}_{{0,1,1}}
+\,a^{(3)}_{{0,0,2}} 
\right) \varepsilon_{\nB}
+ \left(
-{\frac {11\,}{15}}a^{(2)}_{{2,0,0}}
+\frac{1}{6}\,a^{(3)}_{{1,1,0}} 
\right) \varepsilon_{\nB}\varepsilon_{\nC}
\\&&
+\frac{2}{15} \left( 
\,a^{(1)}_{{1,1,0}}
-2\,a^{(2)}_{{0,2,0}}
-2\,a^{(2)}_{{1,0,1}}
+\,a^{(3)}_{{0,1,1}} 
\right) \varepsilon_{{\nA}}\varepsilon_{\nB}
+\frac{1}{30} \left( 
\,a^{(1)}_{{0,2,0}}
+\,a^{(1)}_{{1,0,1}}
-2\,a^{(2)}_{{0,1,1}}
+\,a^{(3)}_{{0,0,2}} 
\right) \varepsilon_{\nA}^{2}
\\&&
+ \frac{1}{36}\left( 
-\,a^{(1)}_{{2,0,0}}
-10\,a^{(2)}_{{1,1,0}}
-13\,a^{(3)}_{{0,2,0}}
+11\,a^{(3)}_{{1,0,1}}
\right) \varepsilon_{\nA}\varepsilon_{\nC}
+\frac{1}{5} \big( 
\,a^{(1)}_{{2,0,0}}
-2 \,a^{(2)}_{{1,1,0}}
+\,a^{(3)}_{{0,2,0}}
\\&&
+\,a^{(3)}_{{1,0,1}} 
\big) {\varepsilon_{{\nB} }}^{2}
+\frac{1}{36}\,a^{(3)}_{{2,0,0}}\varepsilon_{\nA}\varepsilon_{\nB}\varepsilon_{\nC}
+a^{(3)}_{{2,0,0}}{ \varepsilon_{\nB}}^{3}
+\frac {7}{90} \left( 
{\,a^{(1)}_{{2,0 ,0}}}
-{2\,a^{(2)}_{{1,1,0}}}
+{\,a^{(3)}_{{0,2,0}}}
+\,a^{(3)}_{{1,0,1}} 
\right) \varepsilon_{\nA}^{2}\varepsilon_{\nB}
\\&&
+ \frac{1}{120}\left( 
{a^{(1)}_{{1,1,0}}}
-2a^{(2)}_{{0,2,0}}
-2a^{(2)}_{{1,0,1}}
+a^{(3)}_{{0,1,1}} 
\right) \varepsilon_{\nA}^{3}
+ \left( 
-{\frac {7}{180}\,a^{(2)}_{{2,0,0}}}
+{\frac {11}{72} \,a^{(3)}_{{1,1,0}}} 
\right)\varepsilon_{\nA}^{2} \varepsilon_{\nC}
\\&&
+ \frac {11}{15}\left( 
-{\,a^{(2)}_{{2,0,0}}}
+{\frac {1}{2}\,a^{(3)}_{{1,1,0}}}
\right) \varepsilon_{\nA}\varepsilon_{\nB}^{2}
+ \frac{1}{9}\left( 
-\,a^{(2)}_{{2,0,0}}
+\frac{1}{2}\,a^{(3)}_{{1,1,0}}
\right) \varepsilon_{\nB}\varepsilon_{\nA}^{3}+{\frac {11}{72}\,a^{(3)}_{{2,0,0}} \varepsilon_{\nA}^{3}}\varepsilon_{\nC}
\\&&
+{\frac {34}{45}\,a^{(3)}_{{2,0,0}}\varepsilon_{\nA}^{2}}{ \varepsilon_{\nB}}^{2}
%
+\frac{1}{360} \left( 
a^{(1)}_{2,0,0} 
-2a^{(2)}_{1,1,0}
+a^{(3)}_{{0,2,0}}
+a^{(3)}_{{1,0,1}}
\right) \varepsilon_{\nA}^{4}
+\frac{1}{720} \left( 
-2a^{(2)}_{{2,0 ,0}}
+a^{(3)}_{{1,1,0}}
\right) \varepsilon_{\nA}^{5}
\\&&
+{\frac {5}{72}\,a^{(3)}_{{2,0,0}} \varepsilon_{\nB}\varepsilon_{\nA}^{4}}
+{\frac{1}{720} {a^{(3)}_{{2,0,0}}\varepsilon_{\nA}^{6}}}.
\eas
\end{thm}
\bpr 
In order to find the transformation  the following  linear system should be solved
\ba\label{3s}
\ad({\nB}_0^{-1})(\ad(\bar{\nX}_0^\varepsilon){\nt}^\varepsilon_{\nB}-\nX_2)=0,
\ea
where 
\bas
{\nt}^\varepsilon_{1}&=&\sum_{j=0}^{2} \alpha_j \mathsf{A}^j_{1,0}+\sum_{j=-1}^{3} \beta_j{\nB}^{j}_{1, 0}+\sum^{0}_{j=-2} \gamma^j_{ 1} {\nC}^{j}_{-1, 1}+\sum^{4}_{j=-2} \gamma^j_{ 0} {\nC}^{j}_{1, 0},
\eas
or in the different basis it  equals to 
\ba\label{3t}
{\nt}^\varepsilon_{1}&=&\sum_{i+j+k=2}  \alpha^{(1)}_{{i,j,k}} x_1^i x_2^jx_3^k \frac{\partial}{\partial x_1}+\sum_{i+j+k=2}  \alpha^{(2)}_{{i,j,k}} x_1^i x_2^jx_3^k \frac{\partial}{\partial x_1}
\\\nonumber
&+&\sum_{i+j+k=2}  \alpha^{(3)}_{{i,j,k}} x_1^i x_2^jx_3^k \frac{\partial}{\partial x_3}.
\ea
By solving Equation \eqref{3s} one can find the coefficients of transformation
\({\nt}^\varepsilon_{1}\)  as 
 are  given  in  Appendix \ref{AppA}, see Equation \eqref{3dtc}. On the other hand,
by solving  the equation below
\bas
\bar{a}^{1}_0\mathsf{A}^0_{1,0}++\bar{\nb}^{1}_1 {\nB}^{-1}_{1,0}+\bar{\nc}^0_1 {\nC}^{-2}_{1,0}+\bar{\nc}^0_{-1} {\nC}^{-2}_{-1,1}
=\ad(\bar{\nX}_0^\varepsilon){\nt}^\varepsilon_{\nB}-\nX_2,
\eas
we find the  coefficients of normal form which has four free  parameters as
 \(\alpha^{(3)}_{{2,0,0}},\alpha^{(2)}_{{2,0,0}},\alpha^{(3)}_{{0,2,0}}\)  and \(\alpha^{(2)}_{{2,0
		,0}}.\)  In accordance to  the \(\Sl\)-costyle  we can take   all of them zero and we get the coefficients as   given in the theorem.
	These coefficients with  those free parameters are given in Appendix  \ref{AppA}, see equations \eqref{f1}-\eqref{f4}.
\epr
\section{An example on \({\mathfrak{sp}}(4,\mathbb{R})\) }\label{sec:exampsp4}
 
 In \cite[Equation (48)]{mailybaev2001transformation} the versal deformation problem is studied using formal power series.
We refer to this paper for more references to the literature and a general introduction of the importance of the versal deformation in applied mathematics.
We mention that to keep things simple, we use an almost symplectic map to obtain the versal normal form, a trade off we have discussed in Section \ref{sec:2}.

In this  section we  shall find the  near identity transformation \(T^\varepsilon\) as discussed in the previous section to bring the symplectic matrices
 given by \cite[Equation (48)]{mailybaev2001transformation}, describing oscillations of a simply supported elastic pipe
 conveying fluid, to its versal  normal form. 
 Set 
 \bas
 \rho:=\frac{1}{4}\sqrt { \left( 4+\varepsilon\,p_{{1}}\right)  \left( {3}+4\varepsilon\,
 	p_{{2}} \right) },
 \eas	
 where \(p_1,p_2\) are two real parameters.	 Define	
 \ba\label{Russ}
 X^\varepsilon_0:=\begin{pmatrix} 0&{\rho}&1&0\\ \noalign{\medskip}-{\rho}&0&0&1\\ \noalign{\medskip}\varepsilon\,p_{{1}}-{{\rho}}^{2}+3&0&0&{
 	\rho}\\ \noalign{\medskip}0&4\,\varepsilon\,p_{{1}}-{{\rho}}^{2}&-{
 	\rho}&0\end{pmatrix}. 
 \ea	
 Set
 \(
 r:=\frac{\sqrt{3}}{2}\) and define  \(\mathsf{n}_0=X^0_0\) and apply the Jacobson-Morozov  construction to find 
 \bas
 \mathsf{n}_0:= \begin{pmatrix} 0&r&1&0\\ \noalign{\medskip}-r&0&0&1
 \\ \noalign{\medskip}\frac{9}{4}&0&0&r\\ \noalign{\medskip}0&-\frac{3}{4}&-r&0
\end{pmatrix},\quad 
\mathsf{m}_0:=  \begin{pmatrix} 0&-r&1&0\\ \noalign{\medskip}r&0&0&-\frac{1}{3}
 \\ \noalign{\medskip}\frac{9}{4}&0&0&-r\\ \noalign{\medskip}0&\frac{9}{4}&r&0
\end{pmatrix}, 
 \eas	
one can see that   \(
 {\mathsf{m}_0}^4=\mathsf{n}_0^4=0
 .\)
 Now define 
 \bas
 \mathsf{h}_0:=[\mathsf{m}_0,\mathsf{n}_0]=\begin{pmatrix} 0&0&0&-\frac{2}{3}\,\sqrt {3}
 \\ \noalign{\medskip}0&2&-\frac{2}{3}\,\sqrt {3}&0\\ \noalign{\medskip}0&-
 \,\sqrt {3}&0&0\\ \noalign{\medskip}-\frac{3}{2}\,\sqrt {3}&0&0&-2
\end{pmatrix}, 
  \eas	
 and we have 
 \(
 [ \mathsf{h}_0,\mathsf{m}_0]-2\mathsf{m}_0=0,\,\, [\mathsf{h}_0,\mathsf{n}_0]+2\mathsf{n}_0=0
 .\)	
The normal form of \(X^\varepsilon_0\) consists of  those elements which are in  \(\ker{\rm ad}(\mathsf{m}_0)\); in fact 	
 \bas
 V_1=\begin{pmatrix} 0&\frac{\sqrt {3}}{2}&-1&0\\ \noalign{\medskip}-\frac{\sqrt {3}}{2}\,
 &0&0&\frac{1}{3}\,
 	\\ \noalign{\medskip}-\frac{9}{4}&0&0& \frac{\sqrt {3}}{2}
 	\\ \noalign{\medskip}0&-\frac{9}{4}&-\frac{\sqrt {3}}{2}&0
 \end{pmatrix},&\qquad V_2=\begin{pmatrix} 0&\frac{9\sqrt {3}}{4}\,
&-\frac{3}{2}\,
&0\\ \noalign{\medskip}0&0&0&0
 \\ \noalign{\medskip}0&0&0&0 \\ \noalign{\medskip}0&{\frac {81}{4}}&-\frac{9\sqrt {3}}{4}\,&0
\end{pmatrix}.&
 \eas
Hence the normal form is given by 	
 \ba\label{NFRUSS}
\bar{X}_0^\varepsilon=\mathsf{n}_0+\varepsilon_{1} V_1+\varepsilon_{2} V_2.
 \ea
 The characteristic polynomial of \(X^\varepsilon_0\)  given in the \eqref{Russ} and \({\bar X}^\varepsilon_0\) 	 are  as 	
 \bas
\chi(X^\varepsilon_0)&=&{\lambda}^{4}
+ \left( 
-{\frac {17\,\varepsilon\,p_{{1}}}{4}}
+{\varepsilon}^{2}p_{{1}}p_{{2}}+4\,\varepsilon\,p_{{2}} 
\right) {\lambda}^{2}
+4p_{{1}}\varepsilon\left(\,{\varepsilon}{p_{{1}}}
+3\,\right),
\\
 \chi({\bar X}^\varepsilon_0)&=&{\lambda}^{4}+10\,\varepsilon_{1}{\lambda}^{2}+9\,{\varepsilon_{1}}^{2}+54\,\varepsilon_{2}.
 \eas   
 Therefore we obtain     
 \bas                             
 \varepsilon_{1}&=&\frac{1}{5}{\varepsilon}\left({2} p_{{2}}-{\frac {17\,}{8}} p_{{1}}\right)+\frac{1}{10}\,p_{{1}}p_{{2}}\varepsilon^2,
 \\
 \varepsilon_{2}&=&\frac{1}{{86400}}{ {\varepsilon^2\, \left( 29\,p_{{1}}+48\,p_{{2}} \right)  \left( 131
 		\,p_{{1}}-48\,p_{{2}} \right) }}-
 	\frac{1}{{3600}}
 	{ {\varepsilon p_{{1}} \left( 6\,{
 			\varepsilon}^{3}p_{{1}}{p_{{2}}}^{2}-51\,{\varepsilon}^{2}p_{{1}}p_{{2}}+48
 		\,{\varepsilon}^{2}{p_{{2}}}^{2}-800 \right) }}.
 \eas
 With \(\chi(X^\varepsilon_0)=\lambda^4+\Delta_2 \lambda^2+\Delta_4\),
 we see that
 \begin{eqnarray*}
 	\Delta_2&=&-{\frac {17\,}{4}}\varepsilon\,p_{{1}}
 	+{\varepsilon}^{2}p_{{1}}p_{{2}}+4\,\varepsilon\,p_{{2}}, \\
 	\Delta_4&=&4\,{\varepsilon}^{2}{p_{{1}}}^{2}
 	+12\,\varepsilon\,p_{{1}}.
 \end{eqnarray*}
To go from \(p_1, p_2\) to \(\Delta_2,\Delta_4\) is less simple
then in the earlier examples. 
 \begin{thm}\label{TransformationRussia}
	There exists invertible transformation \(	T^\varepsilon\)   such that  brings
	\eqref{Russ} to 
	\eqref{NFRUSS},
 \end{thm}
 \ba\label{transfruss}
 T^{\varepsilon}=\begin{pmatrix} 1&0&0&0
 \\ \noalign{\medskip}0&1+\frac{\sqrt{3}}{76800\rho}\varepsilon t_{1}&-\frac{1}{115200\rho}\varepsilon t_2&0
 \\ \noalign{\medskip}0&{\frac {\sqrt {3}}{76800}} \varepsilon t_3&1+\frac{1}{115200}\varepsilon t_{4}&0
 \\ \noalign{\medskip} \frac{1}{25600 \rho} \varepsilon t_{5}&0&0&1+\frac{\sqrt{3}}{115200\rho}\varepsilon t_{{6}}\end{pmatrix},\ea	                                                                                 
in which 
\bas
t_1&=&320\left(29\,p_{{1}}+48\,p_{{2}}\right)+ \left( 1079\,{p_{{1}}}^{2}+11296\,p_{{1
}}p_{{2}}-2304\,{p_{{2}}}^{2} \right) \varepsilon+8\,p_{{1}}p_{{2}}
\left( 233\,p_{{1}}-144\,p_{{2}} \right) {\varepsilon}^{2}
\\&&
-144\,{
	\varepsilon}^{3}{p_{{2}}}^{2}{p_{{1}}}^{2},
\\
t_2&=&
1920\left(3 p_1+16 p_2\right)+ \left( 1079\,{p_{{1}}}^{2}+15136\,p_{{1
}}p_{{2}}-2304\,{p_{{2}}}^{2} \right) \varepsilon+8\,p_{{1}}p_{{2}}
\left( 233\,p_{{1}}-144\,p_{{2}} \right) {\varepsilon}^{2}
\\&&
-144\,{
	\varepsilon}^{3}{p_{{2}}}^{2}{p_{{1}}}^{2},
\\
t_3&=&
12800\,p_{{1}}+ \left( 6519\,{p_{{1}}}^{2}+2336\,p_{{1}}p_{{2}}-2304
\,{p_{{2}}}^{2} \right) \varepsilon+8\,p_{{1}}p_{{2}} \left( 73\,p_{{1}}-
144\,p_{{2}} \right) {\varepsilon}^{2}-144\,{\varepsilon}^{3}{p_{{2}}}^{2}{p
	_{{1}}}^{2},
\\
t_4&=&
960\left(17 p_1-16 p_2\right)- \left( 6519\,{p_{{1}}}^{2}+6176\,p_{{1}
}p_{{2}}-2304\,{p_{{2}}}^{2} \right) \varepsilon-8\,p_{{1}}p_{{2}}
\left( 73\,p_{{1}}-144\,p_{{2}} \right) {\varepsilon}^{2}
\\&&
+144\,{\varepsilon
}^{3}{p_{{2}}}^{2}{p_{{1}}}^{2},
\\
t_5&=&
320\left(17p_1-16p_2\right)- \left( 331\,{p_{{1}}}^{2}+12704\,p_{{1}}p
_{{2}}-5376\,{p_{{2}}}^{2} \right) \varepsilon-168\,p_{{1}}p_{{2}}
\left( 17\,p_{{1}}-16\,p_{{2}} \right) {\varepsilon}^{2}
\\&&
+336\,{\varepsilon}
^{3}{p_{{2}}}^{2}{p_{{1}}}^{2},
\\
t_6&=&960\left(3 p_{{1}}+16 \,p_{{2}}\right)+ \left( 331\,{p_{{1}}}^{2}+15264\,p_{{1}
}p_{{2}}-5376\,{p_{{2}}}^{2} \right) \varepsilon+168\,p_{{1}}p_{{2}}
\left( 17\,p_{{1}}-16\,p_{{2}} \right) {\varepsilon}^{2}
\\&&
-336\,{\varepsilon}
^{3}{p_{{2}}}^{2}{p_{{1}}}^{2}.
\eas
\bpr
The transformation \eqref{transfruss} is obtained using equation \( {X}^{\varepsilon}_0T^\varepsilon=T^\varepsilon	\bar{X}^{\varepsilon}_0.\)
\epr
 \begin{rem}
In this example, we did not put  the \(X^0_0\) into the symplectic normal form.
 \end{rem} 
 \section{Three body problem}\label{sec:3body}
 \subsection{The versal normal form at $L_4$}\label{sec:versalL4}
 In the theory of the restricted three body problem, the Langrange equilibria play a very practical role,
since they are used to park satellites in orbit, as has been the case for $L_1$ and $L_2$. 
The Trojan points $L_4$ and $L_5$ are considered as positions for space colonies,
since they are stable, unlike $L_3$ which only made it into science fiction sofar.

Consider (Cf. \cite[Equation 1.8]{cushman1986versal}) the four-dimensional  \(L_4\)-singularity
 \ba\label{mec}
 X^\delta_0&:=&X^0_0
 +(4\,\sqrt {2}-3\,\delta)
 \begin{pmatrix} 0&-\frac{{\gamma}^{2}}{4}&\frac{{\gamma}^{2}}{2}&0
 \\ \noalign{\medskip}\frac{1}{8{\gamma}^{2}}&0&0&-\frac{1}{4{\gamma}^{2}}
 \\ \noalign{\medskip}\frac{1}{16{\gamma}^{2}}&0&0&-\frac{1}{8{\gamma}^{2}}
 \\ \noalign{\medskip}0&-\frac{{\gamma}^{2}}{8}&\frac{{\gamma}^{2}}{4}&0
\end{pmatrix},
 \ea
 where  \(\gamma=(1-\frac{\sqrt{2}}{2})^{\frac{1}{2}}\) and
 \bas
 X^0_0:=\begin{pmatrix} 0&-\frac{1}{2}\sqrt {2}&0&0
 \\ \noalign{\medskip}\frac{1}{2}\sqrt {2}&0&0&0\\ \noalign{\medskip}0&0&0&-\frac{1}{2}\sqrt {2}\\ \noalign{\medskip}0&0&\frac{1}{2}\sqrt {2}&0
\end{pmatrix}+\begin{pmatrix} 0&0&0&0
\\ \noalign{\medskip}0&0&0&0\\ \noalign{\medskip}-1&0&0&0\\ \noalign{\medskip}0&-1&0&0
\end{pmatrix}=\s_0+\n_0.
 \eas
The bifurcation point of \(X_0^{\delta}\) is  \(\delta=\delta_0:=\frac{4\sqrt{2}}{3}.\)
 Set \(\delta=\delta_0+\varepsilon\)  to find 
 \ba\label{Aepsilon}
 X^\varepsilon_0=X^0_0+\varepsilon \begin{pmatrix} 0&\frac{3}{4}-\frac{3}{8}\,\sqrt {2}&-\frac{3}{2}+\frac{3}{4}\,\sqrt {2}&0
 \\ \noalign{\medskip}-\frac{3}{4}-\frac{3}{8}\,\sqrt {2}&0&0&\frac{3}{2}+\frac{3}{4}\,\sqrt {2}
 \\ \noalign{\medskip}-\frac{3}{8}-\frac{3}{16}\,\sqrt {2}&0&0&\frac{3}{4}+\frac{3}{8}\,\sqrt {2}
 \\ \noalign{\medskip}0&\frac{3}{8}-\frac{3}{16}\,\sqrt {2}&-\frac{3}{4}+\frac{3}{8}\,\sqrt {2}&0
\end{pmatrix}. 
 \ea

 The normal form of \eqref{mec} is given by 
 \bas
 \bar{X}^\varepsilon_0=X^0_0-{\sqrt{2}}\varepsilon_{\nS}\s_0+\varepsilon_{\nN}\n_0, 
 \eas
 where \(\varepsilon_{\nN},\varepsilon_{\nS}\) are  obtained as follows.
The characteristic polynomial of \(X^\varepsilon_0\) is 
\ba\label{charA}
{\lambda}^{4}+{\lambda}^{2}-\frac{1}{16}\left( 3\,\varepsilon+6+4\,\sqrt {2} \right) 
\left( 3\,\varepsilon-6+4\,\sqrt {2} \right). 
\ea
Comparing the characteristic polynomial of \(X^\varepsilon_0\)  with characteristic polynomial of \( \bar{X}^\varepsilon_0\) as
\ba\label{charB}
{\lambda}^{4}+ \left( 2\,\varepsilon_{\nN}+2{\varepsilon_{\nS}}^{2}-2
\,\sqrt {2}\varepsilon_{\nS}+1 \right) {\lambda}^{2}+\frac{1}{4}\left( -2\,{\varepsilon_{\nS}}^{2}+2\,\sqrt {2}\varepsilon_{\nS}+2
\varepsilon_{\nN}-1 \right) ^{2}.
\ea
Then equations \eqref{charA} and \eqref{charB} imply   that 
\ba\label{FM0}
2\varepsilon_{\nS} \left( \sqrt {2}-\varepsilon_{\nS} \right)&=&2\,
\varepsilon_{\nN},
\\\label{FM1}
\left( 4\,\varepsilon_{\nN}-1 \right) ^{2}&=&-\frac{1}{4}\left( 3\,\varepsilon
-6+4\,\sqrt {2} \right)  \left( 3\,\varepsilon+6+4\,\sqrt {2} \right). 
\ea
 Now, by solving  Equation \eqref{FM1}  respect to \(\varepsilon_{\nN}\)  we find two solutions. The negative root is the right solution, since for \(\varepsilon=0\) we have 
\bas
\varepsilon_{\nN}={\frac {3\,\sqrt {2}}{4}}\varepsilon+{\frac {81}{32}}\varepsilon^2+O \left( {\varepsilon
}^{3} \right), 
\eas 
hence, 
\bas
\varepsilon_{\nN}&=&
-{\frac {\sqrt {- \left( 3\,\varepsilon-6+4\,\sqrt {2} \right) 
			\left( 3\,\varepsilon+6+4\,\sqrt {2} \right) }}{8}}+\frac{1}{4}.
\eas
Now substitute \(\varepsilon_{\nN}\) into \eqref{FM0} and solve for \(\varepsilon_{\nS}\) to get 
\bas
\varepsilon_{\nS}&:=& 
\frac{1}{2}
\left(
{ {\sqrt {2}-\frac{1}{2}\sqrt {4+2\,\sqrt {- \left( 3\,\varepsilon-6+
				4\,\sqrt {2} \right)  \left( 3\,\varepsilon+6+4\,\sqrt {2} \right) }}}}\right).
\eas
Here also we choose  the negative root, since from \(\varepsilon=0\) we obtain 
\bas
\varepsilon_{\nS}=	{\frac {3}{4}}\varepsilon+{\frac {99\,\sqrt {2}}{64}}\varepsilon^2+O \left( {\varepsilon
}^{3} \right). 
\eas 
Note that to have a real normal form we should make this restriction: 
\bas
-3.885618082\approx\frac{-6-4\sqrt{2}}{3}<\varepsilon<\frac{6-4\sqrt{2}}{3}\approx 0.114381918.
\eas
\subsection{Finding the transformation generator \({\nt}^\varepsilon_0\)}
Let \({{\nt}^\varepsilon_0}\in{\rm GL(4,\mathbb{R})}\)
\ba\label{t3b}
{{\nt}^\varepsilon_0}= \frac{1}{\left(\alpha_1+6\right)}\begin{pmatrix} t_{{1}}&t_{{2}}&t_{{3}}&-\frac{1}{\varepsilon}t_{{4}}
\\ \noalign{\medskip}t_{{5}}&t_{{6}}&\frac{1}{\varepsilon}t_{{7}}&t_{{8}}
\\ \noalign{\medskip}-\frac{1}{\varepsilon} t_9&\frac{1}{\varepsilon^2}t_{{10}}&-\frac{1}{\varepsilon^2}t_{{11}}&\frac{1}{\varepsilon}t_{{12}}
\\ \noalign{\medskip}-\frac{1}{\varepsilon^2}t_{13}&\frac{1}{\varepsilon}t_{{14}}&\frac{1}{\varepsilon}t_{{15}}&-\frac{1}{\varepsilon^2}t_{{16}}\end{pmatrix}. 
\ea
We solve the  following equation 
\ba\label{map}
(I+\varepsilon {\nt}^\varepsilon_0){\bar X}_0^\varepsilon=X_0^\varepsilon(I+\varepsilon {\nt}^\varepsilon_0),
\ea
 for \(\{t_i,i=1\cdots16\}.\) The solutions of above equation  respect to four free parameters \(t_1,t_2,t_5,t_6\) are given in  Appendix \ref{AppB}.
 Now we should find  parameters \(t_1,t_2,t_5,t_6.\)	 By substituting parameters in \({\nt}^\varepsilon_0\) and Taylor expansion  around \(\varepsilon=0\) we find 
\bas
{\nt}^\varepsilon_0=\begin{pmatrix} t_1&t_2&0&0\\ \noalign{\medskip}t_5&t_6&0&0
\\ \noalign{\medskip}-{\frac {\sqrt {2} \left( t_{{2}}+t_{{5}} \right) }{3\varepsilon}}&{\frac {\sqrt {2} \left( 4\,t_{{1}}-4\,t_{{6}}-9 \right) }{12
		\varepsilon}}&0&0
\\ \noalign{\medskip}-{\frac {\sqrt {2} \left( 4\,t_{{1}}-4\,t_{{6}}-9 \right) }{12
		\varepsilon}}
&-{\frac {\sqrt {2} \left( t_{{2}}+t_{{5}} \right) }{3\varepsilon}}&0&0\end{pmatrix} +O(\varepsilon^0).
\eas												
{Due to Equation \eqref{map}  transformation \(T^\varepsilon\) should be near identity}.
Hence, it requires 
\(
t_2=-t_5,\,\,t_1=t_{{6}}+\frac{9}{4}.
\)
We have two free parameters \(t_5,t_6\) which can be taken as \(t_5=0, t_6=-\frac{9}{4}.\) 
Hence
\(t_1=t_2=0.\) Thereby we find 
\bas
{\nt}^\varepsilon_0=\begin{pmatrix}0&0&0&0\\ \noalign{\medskip}0&-\frac{9}{4}&0&0
	\\ \noalign{\medskip}0&0&0&0
	\\ \noalign{\medskip}0
	&0&0&0\end{pmatrix} +O(\varepsilon^0).
\eas
\begin{thm}
	The following transformation takes 
		\(X^\delta_0\) to  its versal normal from \(\bar{X}^\varepsilon_0\) through of equation  \( {X}^{\varepsilon}_0(I+{\nt}^\varepsilon_0)=(I+{\nt}^\varepsilon_0)	\bar{X}^{\varepsilon}_0.\)
\bas
{\nt}^\varepsilon_0=\frac{1}{\left( 6+\alpha_1 \right)}\begin{pmatrix} 0&0&0&\frac{1}{8{\varepsilon}}t_{{4}}
\\ \noalign{\medskip}0&\frac{-9}{4}&\frac{1}{6{\varepsilon\,
}}t_{{7}}&0
\\ \noalign{\medskip}0&\frac{1}{24{{\varepsilon}^{2}  }}t_{{10}}&\frac{1}{{96\,{\varepsilon}^{2}
}}t_{{11}}&0
\\ \noalign{\medskip}\frac{1}{24{{\varepsilon}^{2}
}}t_{{13}}&0&0&\frac{1}{24{\varepsilon}^{2}}t_{{16}}\end{pmatrix}, 
\eas
where 
\bas
\alpha_1=\sqrt {- \left( 3\,\varepsilon-6+4\,\sqrt {2} \right)  \left( 3\,\varepsilon
	+6+4\,\sqrt {2} \right) },
\quad
\alpha_2= \sqrt{4+2\alpha_1},
\eas
and 
\bas
t_{{4}}&=&8\left(
2+\sqrt{2}+\alpha_1-\sqrt{2}\alpha_2-\frac{1}{2}\alpha_1\sqrt{2}\right)+
6\left(\frac{3}{2}\,\alpha_1\sqrt{2}-2-\alpha_2-5\sqrt{2}-3\alpha_1
\right)\varepsilon
\\&&
+27\left(\sqrt{2}-2\right){\varepsilon}^{2},
\\
t_{{7}}&=&16\left(\alpha_1
-\sqrt{2}+\frac{1}{2}\alpha_1\sqrt{2}-\sqrt{2}\alpha_2\right)+384\left(3\sqrt{2}\alpha_2-\alpha_2-2\sqrt{2}-4\right)\varepsilon
+27\,\alpha_2{\varepsilon}^{2},
\\
t_{{10}}&=&8\left(
\frac{1}{2}\alpha_1\sqrt{2}\alpha_2-\sqrt{2}\alpha_2-4\sqrt{2}+2\alpha_2+\alpha_1\alpha_2-2\alpha_1\sqrt{2}\right)
+12\big(10+3\alpha_1\sqrt{2}-\sqrt{2}
\alpha_2+\alpha_1
\\&&
+6\sqrt{2}-2\alpha_2\big)\varepsilon
-27\left(\alpha_1-10\right){\varepsilon}^{2},
\\
t_{{11}}&=&16\left(2\,-\alpha_1\right)\alpha_2+\left(48\,\alpha_1\sqrt{2}+36\,\alpha_1\alpha_2-96\,\sqrt{2}
-72\,\alpha_2
-384 \right)\varepsilon
-36\left(
8+4
\,\sqrt{2}+\alpha_2\right){\varepsilon}^{2}
\\&&
+81\,\alpha_2{\varepsilon}^{3},
\\
t_{{13}}&=&
4\,\alpha_2+\sqrt{2}\alpha_2-2\alpha_1\sqrt{2}+2\alpha_1\alpha_2-\alpha_1\alpha_2\sqrt{2}
	+ 12\big(\frac{3}{4}\sqrt{2}\alpha_2\alpha_1-\frac{3}{2}
	\alpha_1\alpha_2-\frac{5}{2}\sqrt{2}\alpha_2-8\,\sqrt{2}
	\\&&
	+\alpha_1-\alpha_2+10
	\big)\varepsilon
+27\left(\sqrt{2}-2\right)\alpha_2{\varepsilon}^{2},
\\
t_{{16}}&=&
4\alpha_2\left(2-\alpha_1\right)
-12
\left(\alpha_1\sqrt{2}
+12
+2\sqrt{2}\right)\varepsilon
+9\left(3\alpha_1\sqrt{2}-6\alpha_1-10\,\sqrt{2}-4-\alpha_2\right){\varepsilon}^{2}
\\&&
+81\left(\sqrt{2}-2\right){\varepsilon}^{3}.
\eas
Note that  \(t_4,t_7\)  are in order \({\varepsilon}\)  and  \(t_{10},t_{11},t_{13}, t_{16}\) are in order \({\varepsilon}^{2}.\) 
\end{thm}
\section{Concluding remarks}
We have shown that the correct implementation of versal normal form in normal form computations is possible.
It does give, and this was to be expected, an added level of complexity.
In any practical computation, this will have to be balanced against the added level of correctness.

It will be interesting to see whether these considerations can also be applied in practice to the theory of unique normal form.
This would, after all, be the holy grail of normal form theory: unique versal normal forms!
\begin{appendices}\label{3dN}
	\addcontentsline{toc}{section}{Appendices}
	\renewcommand{\thesubsection}{\Alph{subsection}}
	
	\section{The coefficients of normal form and transformation of triple-zero }\label{AppA}
	The coefficients of transformation \({\nt}^\varepsilon_{1}\) given in the Equation \eqref{3t}  respect to four free parameters \( \) are as follows:
		\ba\label{3dtc}
	\alpha^{(3)}_{1,1,0}&=&2\,\alpha^{(2)}_{{2,0,0}}+a^{(3)}_{{2,0,0}}+\alpha^{(3)}_{{2,0,0}}\varepsilon_{{\nA}},
	\\\nonumber
	\alpha^{(1)}_{2,0,0}&=&\alpha^{(2)}_{{1,1,0}}-a^{(2)}_{{2,0,0
	}}-\alpha^{(2)}_{{2,0,0}}\varepsilon_{{\nA}}+\alpha^{(3)}_{{2,0,0}}\varepsilon_{{\nB}},
	\\\nonumber
	\alpha^{(1)}_{1,1,0}&=&a^{(1)}_{{2,0,0}}-\alpha^{(3)}_{{2,0,0}}\varepsilon_{{\nC}}-2\,\alpha^{(2)}_{{2,0,0}}\varepsilon_{{\nB}}+
	 \left( \alpha^{(2)}_{{1,1,0}}-a^{(2)}_{{	2,0,0}} \right) \varepsilon_{{\nA}}+
	\alpha^{(3)}_{{2,0,0}}\varepsilon_{{\nB}}\varepsilon_{{\nA}}-\alpha^{(2)}_{{2,0,0}}
	\varepsilon_{{\nA}}^{2},
	\\\nonumber
	\alpha^{(3)}_{1,0,1}&=&\alpha^{(2)}_{{1,1,0}}+\frac{1}{2}\,a^{(3)}_{{1,1,0}}-
	\alpha^{(3)}_{{0,2,0}}+ 2\,\alpha^{(3)}_{{2,0,0}}\varepsilon_{{\nB}}+\left( \alpha^{(2)}_{{2,0,0}}+\frac{1}{2}\,a^{(3)}_{{2,0,0}} \right) 
	\varepsilon_{{\nA}}+\frac{1}{2}\,\alpha^{(3)}_{{2,0,0}} \varepsilon_{{\nA}}^{2},
	\\\nonumber
	3\alpha^{(2)}_{0,2,0}&=&a^{(3)}_{{1,0,1}}+a^{(1)}_{{2,0,0}}+a^{(2)}
	_{{1,1,0}}-\frac{1}{2}\,a^{(3)}_{{0,2,0}}
	+\left( 
	3\alpha^{(2)}_{{1,1,0}}-a^{(2)}_{{2,0,0}}+\frac{1}{2}\,a^{(3)}_{{1,1,0}}-\frac{1}{2}\alpha^{(3)}_{{0,2,0}} \right) 
	\varepsilon_{{\nA}}
	\\\nonumber&&
	+
	\alpha^{(3)}_{{2,0,0}}\varepsilon_{{\nC}}-a^{(3)}_{{2,0,0}}\varepsilon_{{\nB}}+\frac{1}{2}\,a^{(3)}_{{2,0,0}}\varepsilon_{{\nA}}^{2}+\frac{2}{3}
	\,\alpha^{(3)}_{{2,0,0}}\varepsilon_{{\nA}}\varepsilon_{{\nB}}+\frac{1}{2}\,\alpha^{(3)}_{{2,0,0
	}}\varepsilon_{{\nA}}^{3},
	\\\nonumber
	3\alpha^{(3)}_{0,1,1}&=&a^{(1)}_{{2,0,0}}+a^{(2)}_{{1,1,0}}+a^{(3)}_{{0,2,0}}+
	\,a^{(3)}_{{1,0,1}}+\alpha^{(3)}_{{2,0,0}}\varepsilon_{{\nC}}
	+ \left( \frac{1}{2}\,a^{(3)}_{{1,1,0}}-a^{(2)}_{{2,0,0}}+3\alpha^{(2)}_{{1,1,0}
	} \right) \varepsilon_{{\nA}}
	\\\nonumber&&
	+2 \left(a^{(3)}_{{2,0,0}}+3\alpha^{(2)}_
	{{2,0,0}} \right) \varepsilon_{{\nB}}+{5}\alpha^{(3)}_{{2,0,0}}	\varepsilon_{{\nA}}\varepsilon_{{\nB}}
	+\frac{1}{2}\,a^{(3)}_{{2,0,0}}\varepsilon_{{\nA}}^{2}+\frac{1}{2}\,
	\alpha^{(3)}_{{2,0,0}}\varepsilon_{{\nA}}^{3},
	\\\nonumber
	6\alpha^{(2)}_{1,0,1}&=&
	a^{(1)}_{{2,0,0}}+
	a^{(2)}_{{1,1,0}}+ a^{(3)}_{{0,2,0}}-2a^{(3)}_{{1,0,1}}-5\alpha^{(3)}_{{2,0,0}}\varepsilon_{{\nC}}-a^{(3)}_{{2,0,0}}\varepsilon_{{\nB}}
	+ \left( -a^{(2)}_{{2,0,0}
	}-a^{(3)}_{{1,1,0}}+3 \alpha^{(3)}_{{0,2,0}} \right) \varepsilon_{{\nA}}
	\\\nonumber&&
	-4\alpha^{(3)}_{{2,0,0}}\varepsilon_{{\nA}}\varepsilon_{{\nB}}+
	\left( -3 \alpha^{(2)}_{{2,0,0}}- a^{(3)}_{{2,0,0}} \right) {\varepsilon_{{\nA}}
	}^{2}-\varepsilon_{{\nA}}^{3}\alpha^{(3)}_{{2,0,0}},
\\\nonumber
	3\alpha^{(2)}_{0,1,1}&=&\frac{1}{2}\,a^{(1)}_{{1,1,0}}+a^{(2)}_{{0,2,0}}+a^{(2)}_{{1,0,1}}+\left( -\frac{1}{2}\,a^{(3)}_{{2,0,0}}+\alpha^{(2)}_{{2,0,0}} \right) \varepsilon_{{\nC}}+ \left( -\frac{1}{2}\,a^{(3)}_{{1,
			1,0}}+\alpha^{(2)}_{{1,1,0}}-{2}\,a^{(2)}_{{2,0,0}} \right) \varepsilon_{{\nB}}
	\\\nonumber&&
	+ \left( a^{(1)}_{{2,0,0}}+\frac{1}{2}\,a^{(2)}_{{1,1,0}} \right) \varepsilon_{{\nA}}- \left( 
	a^{(3)}_{{2,0,0}}
	+4 \alpha^{(2)}_{{2,0,0}} \right) \varepsilon_{{\nA}}\varepsilon_{{\nB}}-
	\frac{3}{2}\,\alpha^{(3)}_{{2,0,0}}\varepsilon_{{\nA}}\varepsilon_{{\nC}}+
	\left( -a^{(2)}_{{
			2,0,0}}+\frac{3}{2}\,\alpha^{(2)}_{{1,1,0}} \right) \varepsilon_{{\nA}}^{2}
	\\\nonumber&&\nonumber
	-\alpha^{(2)}_{
		{2,0,0}}\varepsilon_{{\nA}}^{3},
\\\nonumber
	3\alpha^{(1)}_{0,2,0}&=&a^{(1)}_{
		{1,1,0}}+{2}a^{(2)}_{{1,0,1}}-a^{(2)}_{{0,2,0}}+\left( -a^{(3)}_{{2,0,0}}+{2}\alpha^{(2)}_{{2,0,0}} \right) \varepsilon_{{\nC}}+ \frac{1}{3}\left( \alpha^{(3)}_{{0,2
			,0}}-a^{(3)}_{{1,1,0}}-4 a^{(2)}_{{2,0,0}} \right) \varepsilon_{{\nB}}
	\\\nonumber&&
	+ \left( -a^{(3)}_{{1,0,1}}+\frac{1}{2}a^{(3)}_{{0,2,0}}+a^{(1)}_{{2,0,0}} \right) \varepsilon_{{\nA}}- \left( 
	{8}\alpha^{(2)}_{{2,0,0}}+a^{(3)}_{{2,0,0}} \right) \varepsilon_{{\nA}}\varepsilon_{{\nB}}-
	4 \alpha^{(3)}_{{2,0,0}}\varepsilon_{{\nA}}\varepsilon_{{\nC}}
	\\\nonumber&&
	+ \left( \frac{3}{2}\,
	\alpha^{(3)}_{{0,2,0}}-\frac{1}{2}a^{(3)}_{{1,1,0}}-a^{(2)}_{{2,0,0}} \right) \varepsilon_{
		{3}}^{2}-{2}\alpha^{(3)}_{{2,0,0}}\varepsilon_{{\nA}}^{2}\varepsilon_{{\nB}}
	- \left( \frac{1}{2}\,a^{(3)}_{{2,0,0}}+{2}\alpha^{(2)}_{{2,0,0}} \right) 
	\varepsilon_{{\nA}}^{3}-\frac{1}{2}\,\varepsilon_{{\nA}}^{4}\alpha^{(3)}_{{2,0,0}},
	\ea
	\bas
	12\alpha^{(3)}_{{0,0,2}}&=&a^{(1)}_{{1,1,0}}+3a^{(3)}_{{0,1,1}}+2 a^{(2)}_{{1,0,1}}+2 a^{(2)}_{{0,2,0}}+ \left( 2 a^{(3)}_{{2,0,0}}+{8}
	\alpha^{(2)}_{{2,0,0}} \right) \varepsilon_{{\nC}}+\alpha^{(3)}_{{2,0,0}}\varepsilon_{{\nB}}^{2}+\alpha^{(3)}_{{2,0
			,0}}\varepsilon_{{\nA}}\varepsilon_{{\nC}}
	\\&&
	+\left(12 \alpha^{(2)}_{{1,1,0}}+2a^{(3)}_{{1,1,0}}-{4}a^{(2)}_{{2
			,0,0}} \right) \varepsilon_{{\nB}}+ \left( a^{(3)}_{{1,0,1}}+a^{(3)}_{{0,2,0}}+2 a^{(2)}_{{1,1,0}}+3 a^{(1)}_{{2,0,0}} \right) \varepsilon_{{\nA}}
	\\&&
	+
	\left(4 \alpha^{(2)}_{{2,0,0}}+3 a^{(3)}_{{2,0
			,0}} \right) \varepsilon_{{\nA}}\varepsilon_{{\nB}}+
	\left(2a^{(3)}_{{1,1,0}}-3a^{(2)}_{{2,0,0}}+6\alpha^{(2)}_{{1,1,0}}
	\right) \varepsilon_{{\nA}}^{2}
	+2 \left( a^{(3)}_{{2,0,0}}-\alpha^{(2)}_{{2
			,0,0}} \right) \varepsilon_{{\nA}}^{3}
	\\&&
	+8\alpha^{(3)}_{{2,0,0}}
\varepsilon_{{\nB}}	\varepsilon_{{\nA}}^{2}+\alpha^{(3)}_{
		{2,0,0}}\varepsilon_{{\nA}}^{4},
	\eas
	\bas
	15\alpha^{(1)}_{0,1,1}&=&\frac{3}{2}\left({3} a^{(1)}_{{0,2,0}}+{3}a^{(1)}_{{1,0,1}}-2a^{(3)}_{{0
			,0,2}}-a^{(2)}_{{0,1,1}}\right)
	-\left( a^{(2)}_{{0,2,0}}+a^{(2)}_{{1,0,1}}-
	{\frac {7}{4}}a^{(1)}_{{1,1,0}}+\frac{3}{4}\,a^{(3)}_{{0,1,1}} \right) \varepsilon_{{\nA}}
	\\&&
	+ \frac{1}{2}\left( 
	-a^{(3)}_{{0,2,0}}+{ {14\,a^{(1)}_{{2,0,0}}}}-{ {13\,a^{(2)}_{{1,1,0}
	}}}-\,a^{(3)}_{{1,0,1}} \right) \varepsilon_{{\nB}}+3 \left( \alpha^{(3)}_{{0,2,0}}-\frac{5}{4}\,a^{(3)}_{{1,1,0}}-3a^{(2)}_{{2,0,0}}
	\right) \varepsilon_{{\nC}}
	\\&&
	+15
	\left( -2\,\alpha^{(2)}_{{2,0,0}}+\frac{1}{3}\,a^{(3)}_{{2,0,0}} \right) \varepsilon_{{\nB}}^
	{2}-15\left( {\frac {5\,a^{(3)}_{{2,0,0}}}{12}}
	+{\frac {
			47\,\alpha^{(2)}_{{2,0,0}}}{30}} \right) \varepsilon_{{\nA}}\varepsilon_{{\nC}}-2 \left( { {7
			a^{(2)}_{{2,0,0}}}}+a^{(3)}_{{1,1,0}} \right) \varepsilon_{{\nA}}\varepsilon_{{\nB
	}}
	\\&&
	+
	\frac{1}{4}\left( 4a^{(1)}_{{2,0,0}}-3a^{(2)}_{{1,1,0}}
	-{ {a^{(3)}_{{1,0,1}}}}-
	{ {a^{(3)}_{{0,2,0}}}} \right) \varepsilon_{{\nA}}^{2}
	-{{14\,\alpha^{(3)}_{{2,0,0}}\varepsilon_{{\nC}}\varepsilon_{{\nB}}}}
	-{\frac 
		{31}{{4}}}\,\alpha^{(3)}_{{2,0,0}}\varepsilon_{{\nC}}\varepsilon_{{\nA}}^{2}
	+{17\,\alpha^{(3)}_{{2,0,0}}\varepsilon_{{\nB}}^{2}}\varepsilon_{{\nA}}
	\\&&
	- \frac{1}{2}\left( a^{(3)}_{{2,0,0}}+10\alpha^{(2)}_{{2,0,0}} \right)
\varepsilon_{{\nB}}	\varepsilon_{
		{3}}^{2}
	- \left( a^{(2)}_{{2,0,0}}
	+{\frac {a^{(3)}_{{1,1,0}}}{8}} \right) \varepsilon_{{\nA}}^{3}-
	\frac{1}{4}\left( {5}\alpha^{(2)}_{{2,0,0}}
	+{\frac {1}{2}}a^{(3)}_{{2,0,0}} \right) 
	\varepsilon_{{\nA}}^{4}
	-{\frac {\alpha^{(3)}_{{2,0,0}}}{8}\varepsilon_{{\nA}}^{5},
	}
	\\
	6\alpha^{(1)}_{1,0,1}&=&a^{(1)}_{{1,1,0}}-4a^{(2)}_{{1,0,1}}+2a^{(2)}_{{0,2,0}}+\left(2a^{(3)}_{{1,0,1}}-a^{(3)}_{{0,2,0
	}}+a^{(1)}_{{2,0,0}} \right) \varepsilon_{{\nA}}- \left( a^{(3)}_{{2,0,0}}+{10}
	\,\alpha^{(2)}_{{2,0,0}} \right) \varepsilon_{{\nC}}
	\\&&
	+6 \left( -\alpha^{(3)}_{{0,2,0}}+\alpha^{(2)}_{{1,1,0}}+\frac{1}{3}\,a^{(3)}_
	{{1,1,0}}-\frac{2}{3}\,a^{(2)}_{{2,0,0}} \right) \varepsilon_{{\nB}}
	+12\,\alpha^{(3)}_{{2,0,0}}\varepsilon_{{\nB}}^{2}+{2}\,\alpha^{(3)}_{{2
			,0,0}}\varepsilon_{{\nA}}\varepsilon_{{\nC}}
	\\&&
	- 2\left( \alpha^{(2)}_
	{{2,0,0}}-a^{(3)}_{{2,0,0}} \right)\varepsilon_{{\nA}}\varepsilon_{{\nB}} 
	+ \left( -3\alpha^{(3)}_{{0,2,0
	}}a^{(3)}_{{1,1,0}}
	-a^{(2)}_{{2,0,0}}+3\alpha^{(2)}_{{1,1,0}} \right) {
		\varepsilon_{{\nA}}}^{2}+7
	\alpha^{(3)}_{{2,0,0}} \varepsilon_{{\nA}}^{2}\varepsilon_{{\nB}}
	\\&&
	+\left( a^{(3)}_{{2,0,0}}+\alpha^{(2)}_{{2,0,0}}
	\right) \varepsilon_{{\nA}}^{3}+\varepsilon_{{\nA}}^{4}\alpha^{(3)}_{{2,0,0}
	},
	\eas
	\bas
	60\alpha^{(2)}_{0,0,2}&=&3\left(a^{(1)}_{{0
			,2,0}}+a^{(1)}_{{1,0,1}}+{{3a^{(2)}_{{0,1,1}}}}-4a^{(3)}_{{0,0,2}}\right)
	+
	60\left( -\frac{1}{10}a^{(2)}_{{2,0,0}}
	-\frac{1}{8}a^{(3)}_{{1,1,0}}+\frac{1}{5}\,\alpha^{(3)}_{{0,2,0}}
	\right) \varepsilon_{{\nC}}
	\\&&
	+
	\left( 8\,a^{(1)}_{{2,0,0}}-{{a^{(2)}_{{1,1,0}}}}-{ {7\,a^{(3)}_{{0,2
					,0}}}}-{ {7\,a^{(3)}_{{1,0,1}}}} \right) \varepsilon_{{\nB}}
	+ \left( 2a^{(1)}_{{1,1,0}}+{ {a^{(2)}_{{0,2,0}}}}+{ {a^{(2)}_{{1,0,1
	}}}}-3a^{(3)}_{{0,1,1}} \right) \varepsilon_{{\nA}}
	\\&&
	- \left( { {29\,\alpha^{(2)}_{{2,0,0}}}}+{
		\frac{25}{2}a^{(3)}_{{2,0,0}}} \right)\varepsilon_{{\nA}} \varepsilon_{{\nC}}-{46\alpha^{(3)}_{{2,0,0}}\varepsilon_{{\nB}}}\varepsilon_{{\nC}}-60
	\left( \alpha^{(2)}_{{2,0,0}}+\frac{1}{3}\,a^{(3)}_{{2,0,0}} \right) \varepsilon_{{\nB}}^{2}
	\\&&
	-
	\left( {{16 a^{(2)}_{{2,0,0}}}}+{ {7\,a^{(3)}_{{1,1,0}}}}
	\right) \varepsilon_{{\nA}}\varepsilon_{{\nB}}+ \left( \frac{3}{2}a^{(1)}_{{2,0
			,0}}-{ 2{a^{(2)}_{{1,1,0}}}}-{ {a^{(3)}_{{0,2,0}}}}-{{a^{(3)}_{{
					1,0,1}}}} \right) \varepsilon_{{\nA}}^{2}
	\\&&
	-\left( 40\alpha^{(2)}_{{2,0,0}}+{
		\frac {19\,}{2}}a^{(3)}_{{2,0,0}} \right) \varepsilon_{{\nA}}^{2}\varepsilon_{{\nB
	}}-{10}\alpha^{(3)}_{{2,0,0}}\varepsilon_{{\nA}}^{3}\varepsilon_{{\nB}}-{
		{8 \alpha^{(3)}_{{2,0,0}}\varepsilon_{{\nA}}^{2}\varepsilon_{{\nC}}}}-{{32}\alpha^{(3)}_{{2,0,0}}\varepsilon_{{\nA}}\varepsilon_{{\nB}}^{2}}
	\\&&
	-\left( \frac{3}{2}a^{(2)}_{{2,0,0}}+
	{2 {a^{(3)}_{{1,1,0}}}} \right) \varepsilon_{{\nA}}^{3}- \left( \frac{5}{2}
	\,\alpha^{(2)}_{{2,0,0}}+{2 {a^{(3)}_{{2,0,0}}}} \right) \varepsilon_{{\nA}}^
	{4}-{2 {\alpha^{(3)}_{{2,0,0}}\varepsilon_{{\nA}}^{5}}},
	\eas
	
	\bas
	\alpha^{(1)}_{0,0,2}&=&\frac{1}{3}\left(\frac{1}{2}\,a^{(1)}_{{0,1,1}}-\,a^{(2)}_{{0,0,2}}\right)+ \frac{1}{30}\left( a^{(3)}
	_{{0,0,2}}+a^{(1)}_{{1,0,1}}-2a^{(2)}_{{0,1,1}}+a^{(1)}_{{0,2,0}}
	\right) \varepsilon_{{\nA}}
	\\&&
	+  \frac{1}{18}\left(\,a^{(1)}_{{2,0,0}}-2\,a^{(3)}_{{0,2,0
	}}-2\,a^{(2)}_{{1,1,0}}+\,a^{(3)}_{{1,0,1}} \right) \varepsilon_{{\nC}}
	+ \frac{1}{12}\left( \,a^{(3)}_{{0,1,1}}+\,a^{(1)}_{{1,1,0}}-2\,a^{(2)}_{{0,2,0}}-2a^{(2)}_{{1,0,1}}
	\right) \varepsilon_{{\nB}}
	\\&&
	+ \left( -
	\frac{1}{5}\,\alpha^{(3)}_{{0,2,0}}-{\frac {11\,}{90}}a^{(2)}_{{2,0,0}}+\frac{1}{36}a^{(3)}_{{1,1,0}}
	\right) \varepsilon_{{\nA}}\varepsilon_{{\nC}}
	+
	\frac{1}{60} \left( 2{ {a^{(1)}_{{1,1,0}}}}+2{ {a^{(3)}_{{0,1,1}}}}-{ {a^{(2)}_{{1,0,1}}}}-{ {a^{(2)}_{
				{0,2,0}}}}
	\right) \varepsilon_{{\nA}}^{2}
	\\&&
	+ \frac{1}{6}\left( -2\,a^{(2)}_{{2,0,0}}+\,a^{(3)}_{{1,1,0}}
	\right) \varepsilon_{{\nB}}^{2}
	+\frac{11}{180} \left( {{a^{(1)}_{{2,0,0}}}}-{
		{2 a^{(2)}_{{1,1,0}}}}+{ {a^{(3)}_{{1,0,1}}}}
	+{ {a^{(3)}_{{0
					,2,0}}}} \right) \varepsilon_{{\nA}}\varepsilon_{{\nB}}
	+\frac{1}{18}\,\alpha^{(3)}_{{2,0,0}}\varepsilon_{{\nC}}^{2}	\\&&-18\left(a^{(3)}_{{2,0,0}}+\frac{1}{18}\alpha^{(2)}_{{2,0,0}} \right) \varepsilon_{{\nB}}\varepsilon_{{\nC}}
	+ \left( \frac{1}{36}\,a^{(3)}_{{2,0,0}}-\frac{1}{10}\,\alpha^{(2)}_{{2,0,0}} \right)\varepsilon_{{\nA}}^{2} \varepsilon_{{\nC}}
	+\frac{17}{360}\left( { {
			a^{(3)}_{{1,1,0}}}}-{2 {a^{(2)}_{{2,0,0}}}} \right)
\varepsilon_{{\nB}}\varepsilon_{{\nA}}^{2}
	\\&&
	+ \frac{1}{360}\left( { {a^{(1)}_{{2,0,0}}}}+{
		{a^{(3)}_{{1,0,1}}}}+{ {a^{(3)}_{{0,2,0}}}}-{2 {a^{(2)}_{{1,1,0
	}}}} \right) \varepsilon_{{\nA}}^{3}
	+{\frac {43\,}{180}}\alpha^{(3)}_{{2,0,0}}\varepsilon_{{\nA}}\varepsilon_{{\nB}}\varepsilon_{{\nC}}
	+\alpha^{(3)}_{{2,0,0}}\varepsilon_{{\nB}}^{3}
	+{\frac {17\,}{36}}a^{(3)}_{{2,0,0}}
	\varepsilon_{{\nA}}	\varepsilon_{{\nB}}^{2}
	\\&&
	+{\frac {11\,}{180}}a^{(3)}_{{2,0,0}}
	\varepsilon_{{\nB}}\varepsilon_{{\nA}}^{3}
	+{\frac {11\,}{360}}\alpha^{(3)}_{{2,0,0}}\varepsilon_{{\nA}}^{3}\varepsilon_{{\nC}}
	+{\frac {34\,}{45}}\alpha^{(3)}_{{2,0
			,0}}\varepsilon_{{\nA}}^{2}\varepsilon_{{\nB}}^{2}
	+{\frac {5\,
		}{72}}\alpha^{(3)}_{{2,0,0}}\varepsilon_{{\nB}}\varepsilon_{{\nA}}^{4}
	+\frac{1}{720} \left( { {a^{(3)}_{{1,1,0}}}}-2{ {a^{(2)}_{{2,0,0}}}} \right) \varepsilon_{{\nA}}^{4}
	\\&&
	+\frac {1}{720}
	a^{(3)}_{{2,0,0}}\varepsilon_{\nA}^{5}+\frac {1}{720}
	\alpha^{(3)}_{{2,0,0}}\varepsilon_{\nA}^{6},
	\eas
	where the  coefficients of the normal form are given by 
	\ba\nonumber
	10\bar{a}^{1}_0&=&a^{(1)}_{{0,2,0}}
	+a^{(1)}_{{1,0,1}}
	+6a^{(3)}_{{0,0,2}}
	+3a^{(2)}_{{0,1,1}}+ \left( -6\alpha^{(3)}_{{0,2,0}}
	+\frac{5}{2}a^{(3)}_{{1,1,0}}-2a^{(2)}_{{2,0,0}}
	+\alpha^{(2)}_{{1,1,0}} \right) \varepsilon_{{\nC}}
	\\&&\nonumber
	+ \left(2a^{(2)}_{{0,2,0}}
	+2a^{(2)}_{{1,0,1}}
	+{\frac {3\,}{2}}a^{(1)}_{{1,1,0}}
	+{\frac {3\,a^{(3)}_{{0,1,1}}}{2}} \right) \varepsilon_{{\nA}}+ \left(a^{(3)}_{{0,2,0}}
	+6 a^{(1)}_{{2,0,0}}
	+{3}a^{(2)}_{{1
			,1,0}}
	+a^{(3)}_{{1,0,1}} \right) \varepsilon_{{\nB}}
	\\\nonumber&&
	+
	\left( {30}a^{(1)}_{{2,0,0}}
	+{ {3\,a^{(2)}_{{1,1,0}}}}
	+a^{(3)}_{{1,0,1
	}}
	+a^{(3)}_{{0,2,0}} \right) \varepsilon_{{\nA}}^{2}
	+8\left(\alpha^{(3)}_{{2,0,0}}\varepsilon_{{\nC}}\varepsilon_{{\nB}}
	+2
	\alpha^{(3)}_{{2,0,0}}\varepsilon_{{\nA}}\varepsilon_{{\nB}}^{2}\right)
	\\\nonumber&&
	+ \left( \frac{5}{2}\,a^{(3)}_{{2,0,0}}
	+{{7\,\alpha^{(2)}_{{2,0,0}}}} \right) 
	\varepsilon_{{\nA}}\varepsilon_{{\nC}}+ \left( -12a^{(2)}_{{2,0,0
	}}
	+a^{(3)}_{{1,1,0}}
	+20 \alpha^{(2)}_{{1,1,0}} \right) \varepsilon_{{\nA}}\varepsilon
	_{{2}}
	+ \left(a^{(3)}_{{2,0,0}}-10\alpha^{(2)}_{{2,0,0}} \right) \varepsilon_{{\nA
	}}^{2}\varepsilon_{{\nB}}
	\\\nonumber&&
	+{2}\alpha^{(3)}_{{2,0,0}
	}\varepsilon_{{\nC}}\varepsilon_{{\nA}}^{2}
	+5\alpha^{(3)}_{{2,0,0}}\varepsilon_{{\nA}}^{3}
	\varepsilon_{{\nB}}
	+ \left( 5\alpha^{(2)}_{
		{1,1,0}}-{3}a^{(2)}_{{2,0,0}}
	+\frac{1}{4}\,a^{(3)}_{{1,1,0}} \right) \varepsilon_{{\nA}}^
	{3}
	+ \left( -\frac{5}{2}\,\alpha^{(2)}_{{2,0,0}}
	+\frac{1}{4}\,a^{(3)}_{{2,0,0}} \right) \varepsilon_{{\nA}}^{4}
	\\&&\label{f1}
	+\frac{1}{4}\,\alpha^{(3)}_{{2,0,0}}\varepsilon_{{\nA}}^{5},
\\\nonumber
	\bar{\nc}^0_{-1}&=&\frac{1}{5}\left(-{2}
	\,a^{(1)}_{{0,2,0}}
	+{3}\,a^{(1)}_{{1,0,1}}-{2}\,a^{(3)}_{{0,0,2}}-a^{(2)}_{{0,1,1}}\right)+\frac{1}{5}\left( a^{(2)}_{{0,2,0}}-
	{4}a^{(2)}_{{1,0,1}}-\frac{1}{2}\,a^{(1)}_{{1,1,0}}-a^{(3)}_{{0
			,1,1}} \right) \varepsilon_{{\nA}}
	\\\nonumber&&
	+ \left( \frac{7}{5}\alpha^{(3)}_{{0,2,0}}-\frac{1}{2}
	\,a^{(3)}_{{1,1,0}}
	-\frac{6}{5}\,a^{(2)}_{{2,0,0}} \right) \varepsilon_{{\nC}}
	- \frac{2}{5}\left( a^{(3)}_{{0,2,0}}+\,a^{(1)}_{{2,0,0}}-\frac{1}{2}\,a^{(2)}_{{1,1,0}}
	+\frac{3}{2}\,a^{(3)}_{
		{1,0,1}} \right) \varepsilon_{{\nB}}
	\\\nonumber&&
	-\left( \frac{1}{2}\,a^{(3)}_
	{{2,0,0}}+\frac {19\,}{5}\alpha^{(2)}_{{2,0,0}} \right) \varepsilon_{{\nC}}
	\varepsilon_{{\nA}}
	+\frac{1}{5} \left( -10\alpha^{(3)}_{{0,2,0}}
	+{4}\,a^{(2)}_{{
			2,0,0}}
	+{3}\,a^{(3)}_{{1,1,0}} \right) \varepsilon_{{\nA}}\varepsilon_{{\nB}}
	\\\nonumber&&
	+ \left( 
	\frac{3}{5}\,a^{(3)}_{{2,0,0}}
	+2\,\alpha^{(2)}_{{2,0,0}} \right) \varepsilon_{{\nA}}^{2}
	\varepsilon_{{\nB}}
	+\alpha^{(3)}_{{2,0,0}}\varepsilon_{{\nA}}^{3}\varepsilon_{{\nB}}+
	\frac{4}{5}\alpha^{(3)}_{{2,0,0}}\varepsilon_{{\nC}}\varepsilon_{{\nB}}
	+\frac{8}{5}\,
	\alpha^{(3)}_{{2,0,0}}\varepsilon_{{\nA}}\varepsilon_{{\nB}}^{2}+\frac{3}{10}\,\alpha^{(3)}_{{2,0,0}}\varepsilon_{{\nA}}^{2}\varepsilon_{{\nC}
	}
	\\\nonumber&&
	+ \frac{1}{5}\left( -\,a^{(1)}_{{2,0,0}}-\frac{1}{2}\,a^{(2)}_{{1,1,0
	}}
	+\frac{3}{10}\,a^{(3)}_{{1,0,1}}-\frac{1}{5}\,a^{(3)}_{{0,2,0}} \right) \varepsilon_{{\nA}}^{2}
	+
	\left( -\frac{1}{2}\,\alpha^{(3)}_{{0,2,0}}
	+\frac{1}{5}\,a^{(2)}_{{2,0,0}}
	+{\frac {3\,a^{(3)}_{{1,1,0}}
		}{20}} \right) \varepsilon_{{\nA}}^{3}
	\\	\label{f4}&&
	+\frac{1}{20}\left(6a^{(1)}_{{2,0,0}}
	+{ {3\,a^{(2)}_{{1,1,0}}}}
	+\,a^{(3)}_{{0,2,0}}
	+
	\,a^{(3)}_{{1,0,1}}
	\right)\varepsilon_{\nA}^2
	+ \left( \frac{1}{2}\,\alpha^{(2)}_{{2,0,0}}
	+{
		\frac {3\,a^{(3)}_{{2,0,0}}}{20}} \right) \varepsilon_{{\nA}}^{4}
	+{\frac {3\,
			\alpha^{(3)}_{{2,0,0}}}{20}}	\varepsilon_{{\nA}}^{5},
	\ea
		\ba
	\nonumber
	\bar{\nc}^0_1&=&
	a^{(1)}_{{0,0,2}}
	+\frac{1}{6} \left( 
	\frac{1}{2}\,a^{(1)}_{{1,1,0}}
	-{5},a^{(2)}_{{1,0,1}}
	+a^{(2)}_{{0,2,0}}
	-\frac{3}{2}\,a^{(3)}_{{0,1,1}} 
	\right) \varepsilon_{{\nC}} 
	+ \frac{1}{5}\left( 
	\,a^{(1)}_{{0,2,0}}
	+a^{(1)}_{{1,0,1}}
	+a^{(3)}_{{0,0,2}}
	-{2}\,a^{(2)}_{{0,1,1}} 
	\right) \varepsilon_{{\nB}}
	\\\nonumber&&
	+ \frac{1}{3}\left( 
	\frac{1}{2}\,a^{(1)}_{{0,1,1}}
	- \,a^{(2)}_{{0,0,2}} 
	\right) \varepsilon_{{\nA}}
	+\frac{1}{5} \left(
	a^{(3)}_{{1,0,1}}
	+a^{(1)}_{{2,0,0}}
	+a^{(3)}_{ {0,2,0}}
	-{2}a^{(2)}_{{1,1,0}} 
	\right) \varepsilon_{{\nB}}^{2}
	\\\nonumber&&
	+\frac{2}{15}
	\left( 
	a^{(1)}_{{1,1,0}}
	-{ {2\,a^{(2)}_{{1,0,1}}}}
	-{ {2\,a^{(2)}_{{0 ,2,0}}}}
	+a^{(3)}_{{0,1,1}} 
	\right) \varepsilon_{\nB}\varepsilon_{\nA}
	+ \left( 
	-\frac{6}{5}\,\alpha^{(3)}_{{0,2,0}}
	+\frac{1}{6}\,a^{(3)}_{{1,1,0}}
	-{ \frac {11\,a^{(2)}_{{2,0,0}}}{15}} 
	\right) \varepsilon_{{\nC}}\varepsilon_{{\nB}}
	\\\nonumber&&
	+ \frac{1}{36}\left( 
	-{ {13}a^{(3)}_{{0 ,2,0}}}
	-a^{(1)}_{{2,0,0}}
	-{ {10\,a^{(2)}_{{1,1,0}}}}
	+{{36} \,a^{(3)}_{{1,0,1}}} 
	\right) \varepsilon_{\nC}\varepsilon_{\nA}
	-\frac{1}{3} \left( 
	\,a^{(3)}_{{2,0,0}}
	+7 \alpha^{(2)}_{{2,0,0}} 
	\right) \varepsilon_ {{1}}^{2}
	\\\nonumber&&
	+ \frac{1}{30}\left( 
	\,a^{(3)}_{{0,0,2}}
	+\,a^{(1)}_{{1,0,1}}
	-2\,a^{(2)}_{{0,1,1}}
	+\,a^{(1)}_{{0,2,0}} 
	\right) \varepsilon_{{\nA}}^{2}
	+ \left( 
	-{ \frac {7\,}{10}}\alpha^{(3)}_{{0,2,0}}
	+{\frac {11\,}{72}}a^{(3)}_{{1,1,0}}
	-{\frac {7\,}{180}}a^{(2)}_{{2,0,0}} 
	\right) \varepsilon_{{\nA}}^{2}\varepsilon_{{\nC}}
	\\\nonumber&&
	+\frac{1}{120} \left( 
	-{2 {a^{(2)}_{{1,0,1}}} }
	-{ 2 {a^{(2)}_{{0,2,0}}}}
	+{{a^{(1)}_{{1,1,0}}}}
	+{ {a^{(3)}_{{0,1 ,1}}}} 
	\right) \varepsilon_{{\nA}}^{3}
	+ \left( 
	\frac{1}{36}\,a^{(3)}_{{2,0,0}}
	-{\frac {34\,}{15} } \alpha^{(2)}_{{2,0,0}}
	\right) \varepsilon_{{\nA}}\varepsilon_{{\nC}}\varepsilon_{{\nB}}
	+a^{(3)}_{{2,0,0}} \varepsilon_{{\nB}}^{3}
	\\\nonumber&&
	+{
		\frac {11}{36}\,\alpha^{(3)}_{{2,0,0}}\varepsilon_{{\nA}}\varepsilon_{{\nC}}^{2}}
	+\frac{8}{5}\,\alpha^{(3)}_{{2,0,0}}\varepsilon_{{\nC}}\varepsilon_{{\nB}}^{2}
	+ \frac{11}{30} \left( 
	-{2{a^{(2)}_{{2,0,0}}}}
	+{ {a^{(3)}_{{1,1,0}}}} 
	\right) \varepsilon_{{\nA}}\varepsilon_{{\nB}}^{2}
	\\\nonumber&&
	+ \frac{7}{90}\left( 
	-{ 2 a^{(2)}_{{1,1,0}}}
	+{ {\,a^{(3)}_{{0,2,0}}}}
	+{ {\,a^{(1)}_{{2,0,0}}}}
	+{ {\,a^{(3)}_{{1,0,1}}}} 
	\right) \varepsilon_{{\nA}}^{2}\varepsilon_{{\nB}}+ \frac{1}{9}\left( 
	\frac{1}{2} \,a^{(3)}_{{1,1,0}}
	-\,a^{(2)}_{{2,0,0}} 
	\right) \varepsilon_{{\nA}}^{3}\varepsilon_{{\nB}}+\frac {16}{5}\,\alpha^{(3)}_{{2,0,0}}\varepsilon_{{\nA}}\varepsilon_{{\nB}}^{3}
	\\\nonumber&&
	+\frac{1}{360} 
	\left(  {a^{(1)}_{{2,0,0}}}
	+{ {a^{(3)}_{{1,0,1}}}}
	+{ { a^{(3)}_{{0,2,0}}}}
	-{2 {a^{(2)}_{{1,1,0}}}} 
	\right) \varepsilon_{{\nA}}^{4}
	+{\frac {34}{45}\,a^{(3)}_{{2,0,0}}\varepsilon_{{\nA}}^{2}\varepsilon_{{\nB}}^{2}}
	+{\frac {49}{45}\,\alpha^{(3)}_{{2,0,0}}\varepsilon_{{\nA}}^{3}\varepsilon_{{\nB}}^{2}}
	+{\frac {5}{72}\,a^{(3)}_{{2,0,0}}\varepsilon_{{\nA}}^{4}\varepsilon_{{\nB}}}
	\\\nonumber&&
	+{\frac {7}{90}\,\alpha^{(3)}_{{2,0,0}}\varepsilon_{{\nA}}^{5}\varepsilon_{{\nB}}}
	+{ \frac {7}{45}\,\alpha^{(3)}_{{2,0,0}}\varepsilon_{{\nA}}^{4}\varepsilon_{{\nC}}}
	+ \left( 
	{\frac {7\,\alpha^{(2)}_{{2,0,0}}}{30}}
	+{\frac {11}{72}\,a^{(3)}_{{2,0,0}}}
	\right) \varepsilon_{{\nA}}^{3}\varepsilon_{{\nC}}
	+{ \frac {34}{45}\,\alpha^{(3)}_{{2,0,0}}\varepsilon_{{\nA}}^{2}\varepsilon_{{\nC}}\varepsilon_{{\nB}}}
	\\\label{f2}&&
	+\frac{1}{360} \left( 
	{\frac {a^{(3)}_{{1,1,0}}}{2}}
	-{ {a^{(2)}_{{2,0,0}} }} 
	\right) \varepsilon_{{\nA}}^{5}
	+\frac{1}{720} {a^{(3)}_{{2,0,0}}\varepsilon_{{\nA}}^{6}}
	+\frac{1}{720} {\alpha^{(3)}_{{2,0,0}}\varepsilon_{{\nA}}^{7}},
	\ea
	\ba
	\nonumber
	 \bar{\nb}^{1}_1 &=&
	\frac{1}{6}\,a^{(1)}_{{0,1,1}}
	+\frac{2}{3}\,a^{(2)}_
	{{0,0,2}}+ \frac{1}{6}\left( -\,a^{(3)}_{{0,1,1}}
	+\,a^{(1)}_{{1,1,0}} \right) \varepsilon_{{\nB}}
	+ \frac{1}{18}\left(4a^{(1)}_{{2,0,0}}
	+\,a^{(3)}_{{0,2,0}}
	+
	\,a^{(2)}_{{1,1,0}}-{ {5\,a^{(3)}_{{1,0,1}}}} \right) \varepsilon_{{\nC}}
	\\\nonumber&&
	+ \frac{1}{12}\left( \,a^{(2)}_{{0,1,1}}
	+\,a^{(1)}_{{0,2,0}}
	+\,a^{(1)}_{{1,0,1}
	}-2\,a^{(3)}_{{0,0,2}} \right) \varepsilon_{{\nA}}
	-\frac{1}{3} \left( {2}a^{(2)}_{{2,0,0}}+\frac{1}{2}\,a^{(3)}_{{1,1,0}} \right) {\varepsilon_{{\nB
	}}}^{2}
	-\frac {7\,}{9}\alpha^{(3)}_{{2,0,0}}\varepsilon_{{\nC}}^{2}
	\\\nonumber&&
	-\frac{1}{3} \left( -\frac{5}{3}a^{(3)}_
	{{2,0,0}}+4\alpha^{(2)}_{{2,0,0}} \right) \varepsilon_{{\nB}}\varepsilon_{{\nC}}
	+\left( \frac{1}{2}\,\alpha^{(3)}_{{0,2,0}}-{\frac {7\,}{18}}a^{(2)}_{{2,0,0}}-{\frac {19
			\,}{72}}a^{(3)}_{{1,1,0}} \right) \varepsilon_{{\nA}}\varepsilon_{{\nC}}
	\\\nonumber&&
	-\frac{4}{3} \left( 2\alpha^{(2)}_{{2,0,0}}+\frac{1}{3}\,a^{(3)}_{{2,0,0}} \right) 
	\varepsilon_{{\nA}}\varepsilon_{{\nB}}^{2}-\frac{4}{9}\,\alpha^{(3)}_{{2
			,0,0}}\varepsilon_{{\nB}}^{2}\varepsilon_{{\nA}}^{2}-{
		\frac {16\,
		}{9}}\alpha^{(3)}_{{2,0,0}}\varepsilon_{{\nA}}\varepsilon_{{\nB}}\varepsilon_{{\nC}}
	-\frac{1}{{12}} \left( {
		{13\,}}\alpha^{(2)}_{{2,0,0}}+\frac {25}{6}a^{(3)}_{{2,0,0}}
	\right) \varepsilon_{{\nA}}^{2}\varepsilon_{{\nC}}
	\\\nonumber&&
	-{\frac {29\,}{72}\alpha^{(3)}_{{2,0,0}}\varepsilon_{{\nA
		}}^{3}\varepsilon_{{\nC}}}
	+ \frac{5}{36}\left( { 2{a^{(1)}_{{2,0,0}}
	}}-{ {a^{(2)}_{{1,1,0}}}}-{{a^{(3)}_{{1,0,1}}}}-{
		{a^{(3)}_{{0,2,0}}}} \right) \varepsilon_{{\nA}}\varepsilon_{{\nB}}+ \frac{1}{24}\left(\,a^{(1)}_{{1,1,0}}-\,a^{(3)}_{{0,1,1}} \right) \varepsilon_{{\nA}}^{2}
	\\\nonumber&&
	-
	\frac{1}{9}\left( \,a^{(3)}_{{1,1,0}}+\,a^{(2)}_{{2,0,0}} \right) \varepsilon_{{\nA}}^{2
	}\varepsilon_{{\nB}}
		+ \frac{1}{72}\left( -{ {a^{(3)}_{{1,0,
					1}}}}
	+2a^{(1)}_{{2,0,0}}-{ {a^{(2)}_{{1,1,0}}}}-{ {a^{(3)}_{{0,2,0
	}}}} \right) \varepsilon_{{\nA}}^{3}
	\\\nonumber&&
-\frac{5}{36}\left( { {a^{(3)}_{{2,0,0}}}}+{6}\alpha^{(2)}_{{2,0
		,0}} \right) \varepsilon_{{\nA}}^{3}\varepsilon_{{\nB}}	-\frac{1}{2} \left(4a^{(2)}_{{2,0,0}}+{
		{a^{(3)}_{{1,1,0}}}} \right) \varepsilon_{{\nA}}^{4}
	- \frac {1}{144}\left( 6
	\alpha^{(2)}_{{2,0,0}}+a^{(3)}_{{2,0,0}} \right) \varepsilon_{{\nA}}^{5
	}
\\\label{f3}&&
-{\frac {5\,}{36}}
\alpha^{(3)}_{{2,0,0}}	\varepsilon_{\nA}^{4}\varepsilon_{{\nB}}-{\frac {1}{144}}\alpha^{(3)}_{{2,0,0}}\varepsilon_{\nA}^{6}.
	\ea
	
\section{The coefficients of   \({\nt}^\varepsilon_0\)   in the $L_4$-problem }\label{AppB}
	The coefficients  of transformation  given in  Equation \eqref{t3b}
with free parameters \(t_1,t_2,t_5,t_6\)   are as follows:
		\bas
		t_{{3}}&=&2 \sqrt{2}\alpha_2t_{{2}}-\alpha_1\sqrt{2}t_{{5}}+4 t_{{5}}+2\alpha_1t_{{5}}+2\sqrt{2}t_{{5}}
		+\frac{4}{3}\left(
		4\,t_{{5}}-2\,\sqrt{2} t_{{5
		}}+\,\alpha_2\,t_{{2}}
		\right)\varepsilon,
		\\
		{t_{{4}}}&=&
		2\left(\sqrt{2}\alpha_2-2-\alpha_1-\frac{\sqrt{2}}{2}\alpha_1\right)
		+\Big(
		\frac{3}{2}\alpha_2-4t_{{6}}-2\alpha_1 t_{
			{6}}-6
		+\sqrt{2}\big(-2t_{{6}}+\alpha_1 t_{{6}}
		\\&&
		+2\alpha_2 t_{{1}}+3
		\big)\Big)\varepsilon
		+\frac{4}{3}\left(2\,\sqrt{2}t_{{6}}
		+\alpha_2t_{{1}}-4t_{{6}}\right)\varepsilon^2,
		\\
		t_{{8}}&=&\alpha_1\sqrt{2}t_{{2}}+2\sqrt{2}\alpha_2t_{{5}}+2\alpha_1t_{{2}}-2\sqrt{2}t_{{2}}+4t_{{2}}
		+\frac{4}{3}\left(2\,\sqrt{2}t_{{2}
		}+\alpha_2t_{{5}}-3 t_{{2}}\right)\varepsilon ,
		\\
		t_{{9}}&=&
		\frac{1}{6}\left( 2\sqrt{2}\alpha_1\alpha_2t_{{2}}+8\,\alpha_1\sqrt{2}t_{{5}}
		+4\,\alpha_1\alpha_2t_{{2}}-4\,\sqrt{2}\alpha_2t_{{2}}+16\,\sqrt{2}t_{{5}}+8\,\alpha_2t_{{2}}\right)
		\\&&
		-\left(\sqrt{2}\alpha_2 t_{{2}}+\alpha_1t_{{5}}+2\alpha_2 t_{{2}
		}+10t_{{5}}\right)\varepsilon,
		\\
		t_{{12}}&=&\frac{1}{3}\left(\,\alpha_1
		-2\right)\alpha_2t_{{5}}+
		\left(2\alpha_1 +\alpha_1\sqrt{2}+2\sqrt{2}-4\right)t_2\varepsilon
		+3\left(\frac{1}{4}\alpha_2t_{{5}}-2t_{{2}}- \sqrt{2}t
		_{{5}}\right)  {\varepsilon}^{2}
		,
		\\
		t_{{15}}&=&\frac{1}{3}\left(\,\alpha_1
		-2\right)\alpha_2t_{{2}}+
		\left(2\alpha_1 \,-\alpha_1\sqrt{2}\,+2\sqrt{2}+4\right)t_5\varepsilon
		+3\left(\frac{1}{4}\alpha_2t_{{2}}+2t_{{5}}- \sqrt{2}t
		_{{5}}\right)  {\varepsilon}^{2},
		\\
		t_{{14}}&=&\frac{2}{3}\left(\frac{\sqrt{2}}{2}\,\alpha_1\alpha_2t_{{5}}-4\sqrt{2}t_{{2}}
		-2\,\alpha_2t_{{5}}-2\,\alpha_1\sqrt{2}t_{{2
		}}-\,\alpha_1\alpha_2t_{{5}}-\,\sqrt{2}\alpha_2t_{{5}}\right)
		+\big(\sqrt{2}\alpha_2\,t_{{5}}+\alpha_1\,t_{{2}}
		\\&&
		-2\,\alpha_2 t_{{
				5}}+10\,t_{{2}}\big)\varepsilon,
		\\
		{t_{{7}}}&=&2\alpha_1-2\sqrt{2}+4+\alpha_1\sqrt{2}-2\sqrt{2}\alpha_2
		+
		\big(\alpha_1\sqrt{2} t_{{1}}-2\sqrt{2}\alpha_2 t_{{6}}+2\alpha_1 t
		_{{1}}-2\sqrt{2}\,t_{{1}}-
		3\sqrt{2}-\frac{3}{2}\,\alpha_2
		\\&&
		+4\,\,t_{{1}}-6\,
		\big)\varepsilon
		-\frac{3}{2}\left(
		2 \sqrt{2}t_{{1}}+\alpha_2t_{{6}}+4\,t_{{1}}\right){\varepsilon}^{2},
		\\
		t_{{10}}&=&
		\frac{4}{3}\left(\frac{1}{4}\alpha_1\sqrt{2}\alpha_2-\alpha_1\sqrt{2}+\frac{1}{2}\alpha_1\alpha_2-\frac{1}{2}\sqrt{2}\alpha_2-2\sqrt{2}+\alpha_2\right)
		+\frac{1}{6}\big(-\sqrt{2}\alpha_2
		-\frac{4}{3}\sqrt{2} t_{{6}}
		\\&&
		+ \frac{1}{3}\alpha_1\sqrt{2}\alpha_2\,t_{{1}}-\frac{4}{3}\alpha_1\sqrt{2}t_{{6}}
		+\frac{2}{3}\alpha_1\alpha_2t_{{1}}-\frac{2}{3}\sqrt{2}\alpha_2 t_{{1}}+\frac{1}{2}\alpha_2t_{{1}}
		-2\alpha_2+\alpha_1+10\,
		\big)\varepsilon
		\\&&
		+\left(-\sqrt{2}\alpha_2t_{{1}}+\alpha_1t_{{6}}-2\alpha_2t_{{1}}+10 t_{{6}}\right){\varepsilon}^{2},
		\\
		t_{{11}}&=&\frac{1}{3}\left(\alpha_1-2\right)\alpha_2
		+\left(2\sqrt{2}+8-\frac{2}{3}\alpha_2
		\,t_{{6}} +\frac{1}{2}\alpha_1\alpha_2t_{{6}}-\alpha_1\sqrt{2}\right)	\varepsilon
		-
		\big(\alpha_1\sqrt{2}t_{{1}}+2\alpha_1t_{{1}}+4t_{{1}}+6
		\\&&
		-2\,\sqrt{2}t_{{1}}
		-3\,\sqrt{2}+\frac{1}{2}\alpha_2
		\big){\varepsilon}^{2}
		+3\left(\sqrt{2} t_{{1}}+\frac{1}{4}\alpha_2t_{{6}}+2t_{
			{1}}\right){\varepsilon}^{3},
		\\
		t_{{13}}&=&\frac{1}{3}\left(8\sqrt{2}-2\sqrt{2}\alpha_2-4\alpha_2+4\alpha_1\sqrt{2}-2\alpha_1\alpha_2+\alpha_1\sqrt{2}\alpha_2\right)
		+
		\big(2\,\alpha_1\sqrt{2}\alpha_2 t_{{6}}-12\,\alpha_2+6\,\sqrt{2}\alpha_2
		\\&&
		+8\,\alpha_1\sqrt{2}t_{{1}}-60\,
		-6\,\alpha_1 -4\,\alpha_1\alpha_2 t_{{6}}
		-4\,\sqrt{2}\alpha_2 t_{{6}}+16\,\sqrt{2} t_{{1}}-8\,\alpha_2 \,t_{{6}}\big)
		\varepsilon+\big(\sqrt{2}\alpha_2t_{{6
		}}-\alpha_1 t_{
			{1}}
		\\&&
		-2\alpha_2t_{{6}}-10\,t_{{1}}\big){\varepsilon}^{2},
		\eas
		\bas
		t_{{16}}&=&	\frac{1}{3}\left(\alpha_1-2\right)\alpha_2+	\left(\frac{1}{3}\alpha_1\alpha_2\,t_{{1}}+\alpha_1\sqrt{2}-\frac{2}{3}\alpha_2	t_{{1}}-2\sqrt{2}+8\right)\varepsilon+
		\big(\alpha_1\sqrt{2}t_{{6}}-2\alpha_1 t_{{6}}-2\sqrt{2}t_{{6}}
		\\&&
		-6+3\sqrt{2}+\frac{3}{4}\alpha_2-4 t_{{6}}\big){\varepsilon}^{2}
		+3\left(\sqrt{2}t_{{6}}+\frac{1}{4}\alpha_2t_{{1}}-2t_{{6}}\right) {\varepsilon}^{3}.
		\eas
			\end{appendices}
\bibliographystyle{plain}

\begin{thebibliography}{10}

\bibitem{arnold1971matrices}
V.I. Arnol'd.
\newblock On matrices depending on parameters.
\newblock {\em Russian Mathematical Surveys}, 26(2):29--43, 1971.

\bibitem{0036-0279-27-5-A02}
V.I. Arnol'd.
\newblock Lectures on bifurcations in versal families.
\newblock {\em Russian Mathematical Surveys}, 27(5):54, 1972.

\bibitem{baez2014𝐺2}
J.~Baez and J.~Huerta.
\newblock $\mathsf{G}_2$ and the rolling ball.
\newblock {\em Transactions of the American Mathematical Society},
  366(10):5257--5293, 2014.

\bibitem{baider1992further}
A.~Baider and J.A. Sanders.
\newblock Further reduction of the {T}akens-{B}ogdanov normal form.
\newblock {\em Journal of Differential Equations}, 99(2):205--244, 1992.

\bibitem{cushman2017uniform}
R.H. Cushman.
\newblock The uniform normal form of a linear mapping.
\newblock {\em Linear Algebra and its Applications}, 512:249--255, 2017.

\bibitem{cushman1986versal}
R.H. Cushman, A.~Kelley, and H.~Ko{\c{c}}ak.
\newblock Versal normal form at the {L}agrange equilibrium ${L}_4$.
\newblock {\em Journal of Differential Equations}, 64(3):340--374, 1986.

\bibitem{cushman56nilpotent}
R.H. Cushman and J.A. Sanders.
\newblock Nilpotent normal forms and representation theory of
  $\mathfrak{sl}_2(\mathbb{R})$, in ‘‘{M}ultiparameter {B}ifurcation
  {T}heory’’ ({M}. {G}olubitsky and {J}. {G}uckenheimer, {E}ds.).
\newblock {\em Contemporary Mathematics}, 56:353--371, 1986.

\bibitem{gazor2017vector}
M~Gazor, F~Mokhtari, and J.A Sanders.
\newblock Vector potential normal form classification for completely integrable
  solenoidal nilpotent singularities.
\newblock {\em arXiv preprint arXiv:1711.09126}, 2017.

\bibitem{MR0161898}
M.~Gerstenhaber.
\newblock The cohomology structure of an associative ring.
\newblock {\em Ann. of Math. (2)}, 78:267--288, 1963.

\bibitem{MR0323842}
J.~Humphreys.
\newblock {\em Introduction to {L}ie {A}lgebras and {R}epresentation {T}heory}.
\newblock Springer-Verlag, New York, 1972.

\bibitem{knapp2013lie}
A.~W. Knapp.
\newblock {\em Lie groups beyond an introduction}, volume 140.
\newblock Springer Science \& Business Media, 2013.

\bibitem{koccak1984normal}
H.~Ko{\c{c}}ak.
\newblock Normal forms and versal deformations of linear {H}amiltonian systems.
\newblock {\em Journal of Differential Equations}, 51(3):359--407, 1984.

\bibitem{Form00}
J.~Kuipers, T.~Ueda, J.A.M. Vermaseren, and J.~Vollinga.
\newblock {FORM} version 4.0.
\newblock {\em Computer Physics Communications}, 184(5):1453--1467, 2013.

\bibitem{lewis2008computer}
R.H. Lewis.
\newblock Computer algebra system {F}ermat, 2008.

\bibitem{mailybaev2001transformation}
A.A. Mailybaev.
\newblock Transformation to versal deformations of matrices.
\newblock {\em Linear Algebra and its Applications}, 337(1-3):87--108, 2001.

\bibitem{SVM2007}
J.~A. Sanders, F.~Verhulst, and J.~Murdock.
\newblock {\em Averaging methods in nonlinear dynamical systems}, volume~59 of
  {\em Applied Mathematical Sciences}.
\newblock Springer, New York, second edition, 2007.

\bibitem{sanders1994versal}
J.A. Sanders.
\newblock Versal normal form computations and representation theory.
\newblock {\em Computer Algebra and Differential Equations’, Cambrige
  University Press, Cambrige}, pages 185--210, 1994.

\bibitem{sanders2003normal}
J.A. Sanders.
\newblock Normal form theory and spectral sequences.
\newblock {\em Journal of Differential Equations}, 192(2):536--552, 2003.

\end{thebibliography}

\end{document}